\documentclass[letterpaper, 12pt]{article}[2000/05/19]
\usepackage[english]{babel}
\usepackage{amsfonts,amsmath,amssymb,amsthm,latexsym,amscd,mathrsfs}
\usepackage{ifthen,cite}
\usepackage[bookmarksnumbered=true]{hyperref}

\hypersetup{pdfpagetransition={Split}}

\newcommand{\evenhead}{AuthorNameForHeading \ name}
\newcommand{\oddhead}{ShortArticleName \ name}
\newcommand{\theArticleName}{Article name}

\newcommand{\FirstPageHeading}[1]{\thispagestyle{empty}%
\noindent\raisebox{0pt}[0pt][0pt]{\makebox[\textwidth]{\protect\footnotesize \sf }}\par}

\newcommand{\ArticleName}[1]{\renewcommand{\theArticleName}{#1}\vspace{-2mm}\par\noindent {\LARGE\bf  #1\par}}
\newcommand{\Author}[1]{\vspace{5mm}\par\noindent {\Large  #1\par} \par\vspace{2mm}\par}
\newcommand{\Address}[1]{\vspace{2mm}\par\noindent {\it #1} \par}
\newcommand{\Abstract}[1]{\vspace{6mm}\par\noindent\hspace*{10mm}
\parbox{140mm}{\small {\bf Abstract.} #1}\par}
\newcommand{\Keywords}[1]{\vspace{3mm}\par\noindent\hspace*{10mm}
\parbox{140mm}{\small {\bf Key words:} \rm #1}\par}
\newcommand{\Classification}[1]{\vspace{3mm}\par\noindent\hspace*{10mm}
\parbox{140mm}{\small {\it 2000 Mathematics Subject Classification:} \rm #1}\vspace{3mm}\par}
\newcommand{\ShortArticleName}[1]{\renewcommand{\oddhead}{#1}}
\newcommand{\AuthorNameForHeading}[1]{\renewcommand{\evenhead}{#1}}

\long\def\@makecaption#1#2{%\vskip\abovecaptionskip
  \sbox\@tempboxa{\small \textbf{#1.}\ \ #2}%
  \ifdim \wd\@tempboxa >\hsize
    {\small \textbf{#1.}\ \ #2}\par \else
    \global \@minipagefalse
    \hb@xt@\hsize{\hfil\box\@tempboxa\hfil}%
  \fi \vskip\belowcaptionskip}

\def\numberwithin#1#2{\@ifundefined{c@#1}{\@nocounterr{#1}}{%
  \@ifundefined{c@#2}{\@nocnterr{#2}}{%
  \@addtoreset{#1}{#2}%
  \toks@\@xp\@xp\@xp{\csname the#1\endcsname}%
  \@xp\xdef\csname the#1\endcsname
    {\@xp\@nx\csname the#2\endcsname.\the\toks@}}}}
\newtheorem{theorem}{Theorem}

{\theoremstyle{definition}

}

\begin{document}

\FirstPageHeading{V.I. Gerasimenko and D.O. Polishchuk}

\ShortArticleName{Evolution of marginal correlation operators}

\AuthorNameForHeading{V.I. Gerasimenko, D.O. Polishchuk}

\ArticleName{On Evolution Equations for Marginal \\ Correlation Operators}

\Author{V.I. Gerasimenko$^\ast$\footnote{E-mail: \emph{gerasym@imath.kiev.ua}}
        and D.O. Polishchuk$^\ast$$^\ast$\footnote{E-mail: \emph{polishuk.denis@gmail.com}}}

\Address{$^\ast$\hspace*{1mm}Institute of Mathematics of NAS of Ukraine,\\
         \hspace*{3mm}3, Tereshchenkivs'ka Str.,\\
         \hspace*{3mm}01601 Kyiv-4, Ukraine}

\Address{$^\ast$$^\ast$Taras Shevchenko National University of Kyiv,\\
    \hspace*{4mm}Department of Mechanics and Mathematics,\\
    \hspace*{4mm}2, Academician Glushkov Av.,\\
    \hspace*{4mm}03187 Kyiv, Ukraine}

\bigskip

\Abstract{
This paper is devoted to the problem of the description of nonequilibrium correlations in quantum
many-particle systems. The nonlinear quantum BBGKY hierarchy for marginal correlation operators is
rigorously derived from the von Neumann hierarchy for correlation operators that give an alternative
approach to the description of states in comparison with the density operators. A nonperturbative
solution of the Cauchy problem of the nonlinear quantum BBGKY hierarchy for marginal correlation
operators is constructed.
}

\bigskip

\vspace{2pc}
\Classification{35Q40; 47J35; 47H20; 82C10; 82C22.}

\makeatletter
\renewcommand{\@evenhead}{
\hspace*{-3pt}\raisebox{-15pt}[\headheight][0pt]{\vbox{\hbox to \textwidth {\thepage \hfil \evenhead}\vskip4pt \hrule}}}
\renewcommand{\@oddhead}{
\hspace*{-3pt}\raisebox{-15pt}[\headheight][0pt]{\vbox{\hbox to \textwidth {\oddhead \hfil \thepage}\vskip4pt \hrule}}}
\renewcommand{\@evenfoot}{}
\renewcommand{\@oddfoot}{}
\makeatother

\newpage
\vphantom{math}
\protect\tableofcontents

\vspace{0.5cm}

\section{Introduction}
The importance of the mathematical description of correlations in numerous problems
of the modern statistical mechanics is well-known. Among them in particular, we refer
to such fundamental problems as the problem of quantum measurements and of a description
of collective behavior of interacting particles by quantum kinetic equations
\cite{AA,BDGM,ESchY2,EShY10,FGS,M1,PP09,S-R,G11,GT}. Owing to the intrinsic complexity
and richness of these problems, primarily it is necessary to develop an adequate mathematical
theory of underlying evolution equations.

The goal of this paper is to derive rigorously the evolution equations for marginal correlation
operators that give an equivalent approach to the description of the evolution of states in
comparison with marginal density operators governed by the quantum BBGKY hierarchy and to
construct a solution of the corresponding Cauchy problem.

We briefly outline the results and structure of the paper.
In introductory section 2 we set forth an approach to the description of the evolution of
states of quantum many-particle systems within the framework of correlation operators governed
by the von Neumann hierarchy \cite{GerShJ,GP}.
In section 3 we introduce the notion of marginal correlation operators. To justify this notion,
we discuss in detail the motivation of the description of states within the framework of marginal
correlation operators or in other words, the origin of the microscopic description of correlations
in quantum many-particle systems. Then we rigorously derive the nonlinear quantum BBGKY hierarchy
for marginal correlation operators from the von Neumann hierarchy for correlation operators. The
nonlinear quantum BBGKY hierarchy gives an alternative method of the description of the evolution
of states of infinitely many particles in comparison with the quantum BBGKY hierarchy for the marginal
density operators \cite{BogLect,Pe95}.
In section 4 we construct the nonperturbative solution of the Cauchy problem of the nonlinear quantum
BBGKY hierarchy. The nonperturbative solution is determined in the form of expansion over particle
clusters, which evolution is governed by the corresponding-order cumulant of the nonlinear group of
operators generated by the solution of the von Neumann hierarchy. The existence theorem for initial
data from the space of trace-class operators is proved. We also give some comments on the mean field
(self-consistent field) asymptotic behavior of the constructed solution.
Finally, in section 5 we conclude with some observations and perspectives for future research in the
light of the results we present here.

%%%%%%%%%%%%%%%%%%%%%%%%%%%%%%%%%%%%%%%%%%%%%%%%%%%%%%%%%%%%%%%%%%%%%%%%%%%%%%%%%%%%%%%%%%%%%%%%%%%%%%%%%%%%%%%%%%%%%%%%%
%%%%%%%%%%%%%%%%%%%%%%%%%%%%%%%%%%%%%%%%%%%%%%%%%%%%%%%%%%%%%%%%%%%%%%%%%%%%%%%%%%%%%%%%%%%%%%%%%%%%%%%%%%%%%%%%%%%%%%%%%

\section{Preliminary facts: the von Neumann hierarchy}
We consider a quantum system of a non-fixed, i.e. arbitrary but finite, number of identical
(spinless) particles with unit mass $m=1$ in the space $\mathbb{R}^{\nu}, \nu\geq 1$, that
obey the Maxwell-Boltzmann statistics. Let
$\mathcal{F}_{\mathcal{H}}={\bigoplus\limits}_{n=0}^{\infty}\mathcal{H}_{n}$ be the Fock space
over the Hilbert space $\mathcal{H}$, where the $n$-particle Hilbert space $\mathcal{H}_n\equiv
\mathcal{H}^{\otimes n}$ is a tensor product of $n$ Hilbert spaces $\mathcal{H}$ and we adopt
the usual convention that $\mathcal{H}_{0}=\mathbb{C}$. The Hamiltonian $H_{n}$ of the $n$-particle
system is a self-adjoint operator with the domain $\mathcal{D}(H_{n})\subset\mathcal{H}_{n}$
\begin{eqnarray*}\label{H_n}
    &&H_{n}=\sum\limits_{i=1}^{n}K(i)+
           \,\sum\limits_{i_{1}<i_{2}=1}^{n}\Phi(i_{1},i_{2}),
\end{eqnarray*}
where $K(i)$ is the operator of a kinetic energy of the $i$ particle and $\Phi(i_{1},i_{2})$ is the
operator of a two-body interaction potential. In particular on functions $\psi_n$ that belong to the
subspace $L^{2}_{0}(\mathbb{R}^{\nu n})\subset \mathcal{D}(H_n) \subset L^{2}(\mathbb{R}^{\nu n})$
of infinitely differentiable symmetric functions with compact supports the operator $K(i)$ acts
according to the the formula: $K(i)\psi_n = -\frac{\hbar^2}{2}\Delta_{q_i}\psi_n$, where $2\pi\hbar$
is a Planck constant, and for the operator $\Phi$ we have: $\Phi(i_{1},i_{2})\psi_{n}=
\Phi(q_{i_{1}},q_{i_{2}})\psi_{n}$, respectively. We assume that the function $\Phi(q_{i_{1}},q_{i_{2}})$
is symmetric with respect to permutations of arguments and it is translation-invariant bounded function.

States of a system of the Maxwell-Boltzmann particles belong to the space $\mathfrak{L}^{1}(\mathcal{F}_\mathcal{H})
=\oplus_{n=0}^{\infty}\mathfrak{L}^{1}(\mathcal{H}_{n})$ of sequences $f=(f_0,f_{1},\ldots,f_{n},\ldots)$
of trace-class operators $f_{n}\equiv f_{n}(1,\ldots,n)\in\mathfrak{L}^{1}(\mathcal{H}_{n})$ and $f_0\in\mathbb{C}$,
that satisfy the symmetry condition: $f_{n}(1,\ldots,n)=f_{n}(i_{1},\ldots,i_{n})$ for arbitrary
$(i_{1},\ldots,i_{n})\in(1,\ldots,n)$, equipped with the norm
\begin{eqnarray*}
    &&\|f\|_{\mathfrak{L}^{1} (\mathcal{F}_\mathcal{H})}=
        \sum\limits_{n=0}^{\infty} \|f_{n}\|_{\mathfrak{L}^{1}(\mathcal{H}_{n})}=
        \sum\limits_{n=0}^{\infty}~\mathrm{Tr}_{1,\ldots,n}|f_{n}(1,\ldots,n)|,
\end{eqnarray*}
where $\mathrm{Tr}_{1,\ldots,n}$ are partial traces over $1,\ldots,n$ particles \cite{Pe95}. We denote
by $\mathfrak{L}^{1}_{0}(\mathcal{F}_\mathcal{H})$ the everywhere dense set in
$\mathfrak{L}^{1}(\mathcal{F}_\mathcal{H})$
of finite sequences of degenerate operators with infinitely differentiable kernels with compact supports.

We describe states of a system by means of sequences $g(t)=(g_0,g_{1}(t,1),\ldots,g_{n}(t,1,\ldots,n)$, $\ldots)\in\mathfrak{L}^{1}(\mathcal{F}_\mathcal{H})$ of the correlation operators $g_{n}(t),\,n\geq1$.
The evolution of all possible states is determined by the initial-value problem of the von Neumann hierarchy
\cite{GerShJ,GP}
\begin{eqnarray}
 \label{vNh}
   &&\frac{d}{dt}g_{s}(t,Y)=\mathcal{N}(Y\mid g(t)),\\ \nonumber\\
 \label{vNhi}
   &&g_{s}(t,Y)\big|_{t=0}=g_{s}(0,Y),\quad s\geq1,
\end{eqnarray}
where the following notations are used:
\begin{eqnarray}\label{vNgenerator}
   &&\mathcal{N}(Y\mid g(t))\doteq-\mathcal{N}_{s}(Y)g_{s}(t,Y)+\\
   &&+\sum\limits_{\mathrm{P}:\,Y=X_{1}\bigcup X_2}\,\sum\limits_{i_{1}\in X_{1}}
      \sum\limits_{i_{2}\in X_{2}}
      (-\mathcal{N}_{\mathrm{int}}(i_{1},i_{2}))g_{|X_{1}|}(t,X_{1})g_{|X_{2}|}(t,X_{2}),\nonumber
\end{eqnarray}
${\sum\limits}_{\mathrm{P}:\,Y=X_{1}\bigcup X_2}$ is the sum over all possible partitions
$\mathrm{P}$ of the set $Y\equiv(1,\ldots,s)$ into two nonempty mutually disjoint subsets
$X_1\subset Y$ and $X_2\subset Y$, the operator $(-\mathcal{N}_{s})$ defined on
$\mathfrak{L}^{1}_0(\mathcal{H}_s)$ by the formula
\begin{eqnarray}\label{komyt}
   &&(-\mathcal{N}_{s}(Y))f_{s}\doteq-\frac{i}{\hbar}(H_{s}f_{s}-f_{s}H_{s}),
\end{eqnarray}
is the generator of the von Neumann equation \cite{DauL} and the operator $(-\mathcal{N}_{\mathrm{int}})$
is defined by
\begin{eqnarray}\label{oper Nint2}
   &&(-\mathcal{N}_{\mathrm{int}}(i_{1},i_{2}))f_{s}\doteq
     -\frac{i}{\hbar}(\Phi(i_{1},i_{2})f_{s}-f_{s}\Phi(i_{1},i_{2})).
\end{eqnarray}

Hereafter we use the following notations: $(\{X_1\},\ldots,\{X_{|\mathrm{P}|}\})$ is a set, elements of
which are $|\mathrm{P}|$ mutually disjoint subsets $X_i\subset Y\equiv(1,\ldots,s)$ of the partition
$\mathrm{P}:Y=\cup_{i=1}^{|\mathrm{P}|}X_i$, i.e. $|(\{X_1\},\ldots,\{X_{|\mathrm{P}|}\})|=|\mathrm{P}|$.
In view of these notations we state that $(\{Y\})$ is the set consisting of one element $Y=(1,\ldots,s)$
of the partition $\mathrm{P}$ $(|\mathrm{P}|=1)$ and $|(\{Y\})|=1$. We introduce the declasterization
mapping $\theta: (\{X_1\},\ldots,\{X_{|\mathrm{P}|}\})\rightarrow Y$, by the following formula:
$\theta(\{X_1\},\ldots,\{X_{|\mathrm{P}|}\})=Y$. For example, let $X\equiv(1,\ldots,s+n)$, then for the
set $(\{Y\},X\setminus Y)$ it holds: $\theta(\{Y\},X\setminus Y)=X$.

On the space $\mathfrak{L}^{1}(\mathcal{H}_n)$ we also introduce the mapping: $\mathbb{R}\ni
t\rightarrow\mathcal{G}_{n}(-t)f_{n}$, which is generated by the solution of the von Neumann
equation of $n$ particles \cite{DauL,BR}
\begin{eqnarray}\label{groupG}
  &&\mathcal{G}_{n}(-t)f_{n}\doteq e^{-{\frac{i}{\hbar}}t H_{n}}\,f_{n}\,e^{{\frac{i}{\hbar}}t H_{n}}.
\end{eqnarray}
This mapping is an isometric strongly continuous group that preserves positivity and self-adjointness
of operators \cite{DauL}. On $\mathfrak{L}_{0}^{1}(\mathcal{H}_{n})\subset\mathcal{D}(-\mathcal{N}_{n})$
the infinitesimal generator of group (\ref{groupG}) is determined by operator (\ref{komyt}).

A solution of the Cauchy problem (\ref{vNh})-(\ref{vNhi}) is given by the following expansion
\cite{GerShJ,GP},\cite{G}
\begin{eqnarray}\label{rozvNh}
  &&g_{s}(t,Y)=\mathcal{G}(t;Y|g(0))\doteq\\
  &&\doteq\sum\limits_{\mathrm{P}:\,Y=\bigcup_i X_i}
     \mathfrak{A}_{|\mathrm{P}|}(t,\{X_1\},\ldots,\{X_{|\mathrm{P}|}\})
     \prod_{X_i\subset \mathrm{P}}g_{|X_i|}(0,X_i),\quad s\geq1,\nonumber
\end{eqnarray}
where ${\sum\limits}_{\mathrm{P}:\,Y=\bigcup_i X_i}$ is the sum over all possible partitions $\mathrm{P}$
of the set $Y\equiv(1,\ldots,s)$ into $|\mathrm{P}|$ nonempty mutually disjoint subsets $X_i\subset Y$,
the evolution operator $\mathfrak{A}_{|\mathrm{P}|}(t)$ is the $|\mathrm{P}|th$-order cumulant of groups
of operators (\ref{groupG}) defined by the formula
\begin{eqnarray} \label{cumulantP}
    &&\mathfrak{A}_{|\mathrm{P}|}(t,\{X_1\},\ldots,\{X_{|\mathrm{P}|}\})\doteq\\
    &&\doteq\sum\limits_{\mathrm{P}^{'}:\,(\{X_1\},\ldots,\{X_{|\mathrm{P}|}\})=
       \bigcup_k Z_k}(-1)^{|\mathrm{P}^{'}|-1}({|\mathrm{P}^{'}|-1})!
       \prod\limits_{Z_k\subset\mathrm{P}^{'}}\mathcal{G}_{|\theta(Z_{k})|}(-t,\theta(Z_{k})).\nonumber
\end{eqnarray}
Here $\sum_{\mathrm{P}^{'}:\,(\{X_1\},\ldots,\{X_{|\mathrm{P}|}\})=\bigcup_k Z_k}$ is the sum over all
possible partitions $\mathrm{P}^{'}$ of the set $(\{X_1\},\ldots,$ $\{X_{|\mathrm{P}|}\})$ into
$|\mathrm{P}^{'}|$ nonempty mutually disjoint subsets $Z_k\subset (\{X_1\},\ldots,$ $\{X_{|\mathrm{P}|}\})$.
For operators \eqref{rozvNh} the estimate holds
\begin{eqnarray}\label{gEstimate}
    &&\big\|g_s(t)\big\|_{\mathfrak{L}^{1}(\mathcal{H}_{s})}\leq s!e^{2s}\emph{c}^s,
\end{eqnarray}
where $\emph{c}\equiv{\max\limits}_{\mathrm{P}:Y=\bigcup_i X_i}(\|g_{|X_1|}(0)\|_{\mathfrak{L}^{1}(\mathcal{H}_{|X_1|})},
\ldots,\|g_{|X_{|\mathrm{P}|}|}(0)\|_{\mathfrak{L}^{1}(\mathcal{H}_{|X_{|\mathrm{P}|}|})})$.

If $g_{n}(0)\in \mathfrak{L}^{1}_{0}(\mathcal{H}_{n})\subset\mathfrak{L}^{1}(\mathcal{H}_{n}),\,n\geq1$,
expansion (\ref{rozvNh}) is a strong (classical) solution of the Cauchy problem (\ref{vNh})-(\ref{vNhi})
and for arbitrary initial data $g_{n}(0)\in\mathfrak{L}^{1}(\mathcal{H}_{n}),\,n\geq1$,
it is a weak (generalized) solution \cite{GerShJ,GP}.

In case of the absence of correlations between particles at initial time, i.e. initial data satisfying
a chaos condition, the sequence of correlation operators has the form
\begin{eqnarray}\label{gChaos}
   &&g(0)=(0,g_{1}(0,1),0,\ldots).
\end{eqnarray}
The corresponding solution of the initial-value problem of the von Neumann hierarchy is given
by the expansion
\begin{eqnarray}\label{gth}
   &&g_{s}(t,Y)=\mathfrak{A}_{s}(t,Y)\,\prod\limits_{i=1}^{s}g_{1}(0,i),
\end{eqnarray}
where $\mathfrak{A}_{s}(t)$ is the $sth$-order cumulant defined by
\begin{eqnarray*}\label{cumcp}
   &&\mathfrak{A}_{s}(t,Y)=\sum\limits_{\mathrm{P}:\,Y=
      \bigcup_i X_i}(-1)^{|\mathrm{P}|-1}({|\mathrm{P}|-1})!
      \prod\limits_{X_i\subset\mathrm{P}}\mathcal{G}_{|X_i|}(-t,X_i).
\end{eqnarray*}
For operators (\ref{gth}) estimate (\ref{gEstimate}) takes the corresponding form:
\begin{eqnarray*}\label{gEstimateh}
   &&\big\|g_s(t)\big\|_{\mathfrak{L}^{1}(\mathcal{H}_{s})}\leq
      s!e^{s}\big\|g_1(0)\big\|_{\mathfrak{L}^{1}(\mathcal{H})}^s.
\end{eqnarray*}

We note that correlations created in evolutionary process of a system are described
by formula (\ref{gth}) and determined by the corresponding-order cumulant of the groups of
operators (\ref{groupG}) of the von Neumann equations.

%%%%%%%%%%%%%%%%%%%%%%%%%%%%%%%%%%%%%%%%%%%%%%%%%%%%%%%%%%%%%%%%%%%%%%%%%%%%%%%%%%%%%%%%%%%%%%%%%%%%%%%%%%%%%%%%%%%
%%%%%%%%%%%%%%%%%%%%%%%%%%%%%%%%%%%%%%%%%%%%%%%%%%%%%%%%%%%%%%%%%%%%%%%%%%%%%%%%%%%%%%%%%%%%%%%%%%%%%%%%%%%%%%%%%%%

\section{The nonlinear quantum BBGKY hierarchy}
The evolution of states of infinite-particle quantum systems is traditionally described by the
marginal (or $s$-particle) density operators governed by the quantum BBGKY hierarchy \cite{BogLect,Pe95}.
In this section we introduce the marginal correlation operators that give an equivalent approach
to the description of the evolution of such states and describe the nonequilibrium correlations in
quantum systems. We also rigorously derive the nonlinear quantum BBGKY hierarchy for marginal correlation
operators from the von Neumann hierarchy (\ref{vNh}) for correlation operators.

\subsection{Marginal correlation operators and marginal density operators}
In the capacity of an example of a mean-value functional of observables \cite{GP} we
consider the definition of the mean-value functional of the additive-type observable
$A^{(1)}=(0,a_{1}(1),\ldots,$ $\sum_{i_{1}=1}^{n}a_1(i_{1}),\ldots)$
\begin{eqnarray}\label{averagegs}
    &&\langle A^{(1)}\rangle(t)=\sum\limits_{n=0}^{\infty}\frac{1}{n!}
        \,\mathrm{Tr}_{1,\ldots,1+n}\,a_{1}(1)g_{1+n}(t,1,\ldots,1+n),
\end{eqnarray}
where the operators $g_{1+n}(t),\,n\geq0$, are determined by expansions \eqref{rozvNh}, and the
functional of the dispersion for this type of observables
\begin{eqnarray}\label{dispg}
    &&\langle(A^{(1)}-\langle A^{(1)}\rangle)^2\rangle(t)=\\
    &&=\sum\limits_{n=0}^{\infty}\frac{1}{n!}\,\mathrm{Tr}_{1,\ldots,1+n}\,(a_1^2(1)-
       \langle A^{(1)}\rangle^2(t))g_{1+n}(t,1,\ldots,1+n)+\nonumber\\
    &&+\sum\limits_{n=0}^{\infty}\frac{1}{n!}
       \,\mathrm{Tr}_{1,\ldots,2+n}\,a_{1}(1)a_{1}(2)g_{2+n}(t,1,\ldots,2+n).\nonumber
\end{eqnarray}
For $A^{(1)}\in\mathfrak{L}(\mathcal{F}_\mathcal{H})$ and $g\in\mathfrak{L}^{1}(\mathcal{F}_\mathcal{H})$
functionals \eqref{averagegs},\eqref{dispg} exists.

Following to formula \eqref{dispg}, we introduce the marginal correlation operators by the series
\begin{eqnarray}\label{Gexpg}
   &&G_{s}(t,1,\ldots,s)\doteq\sum\limits_{n=0}^{\infty}\frac{1}{n!}\,
      \mathrm{Tr}_{s+1,\ldots,s+n}\,\,g_{s+n}(t,1,\ldots,s+n),\quad s\geq1,
\end{eqnarray}
where the sequence $g_{s+n}(t,1,\ldots,s+n),\,n\geq0$, is a solution of the Cauchy problem of the
von Neumann hierarchy (\ref{vNh}). According to estimate (\ref{gEstimate}), series (\ref{Gexpg}) exists
and the estimate holds:
$\big\|G_s(t)\big\|_{\mathfrak{L}^{1}(\mathcal{H}_{s})}\leq s!(2e^2)^s\emph{c}^s\sum_{n=0}^{\infty}(2e^2)^n\emph{c}^n$.
Thus, macroscopic characteristics of fluctuations of observables are determined by marginal correlation
operators \eqref{Gexpg} on the microscopic level
\begin{eqnarray*}
    &&\langle(A^{(1)}-\langle A^{(1)}\rangle)^2\rangle(t)=
      \mathrm{Tr}_{1}\,(a_1^2(1)-\langle A^{(1)}\rangle^2(t))G_{1}(t,1)+
      \mathrm{Tr}_{1,2}\,a_{1}(1)a_{1}(2)G_{2}(t,1,2).
\end{eqnarray*}

Traditionally marginal correlation operators are introduced by means of the cluster expansions of the
marginal density operators $F_{s}(t),\,s\geq1$, governed by the quantum BBGKY hierarchy \cite{BogLect}
\begin{eqnarray}\label{FG}
   &&F_{s}(t,Y)=\sum\limits_{\mbox{\scriptsize $\begin{array}{c}\mathrm{P}:Y=\bigcup_{i}X_{i}\end{array}$}}
      \prod_{X_i\subset \mathrm{P}}G_{|X_i|}(t,X_i),\quad s\geq1,
\end{eqnarray}
where ${\sum\limits}_{\mathrm{P}:Y=\bigcup_{i} X_{i}}$ is the sum over all possible partitions $\mathrm{P}$
of the set $Y\equiv(1,\ldots,s)$ into $|\mathrm{P}|$ nonempty mutually disjoint subsets $X_i\subset Y$.
Hereupon solutions of cluster expansions \eqref{FG}
\begin{eqnarray}\label{gBigfromDFB}
   &&G_{s}(t,Y)=\sum\limits_{\mbox{\scriptsize $\begin{array}{c}\mathrm{P}:Y=\bigcup_{i}X_{i}\end{array}$}}
      (-1)^{|\mathrm{P}|-1}(|\mathrm{P}|-1)!\,\prod_{X_i\subset \mathrm{P}}F_{|X_i|}(t,X_i), \quad s\geq1,
\end{eqnarray}
are interpreted as the operators that describe correlations of many-particle systems. Thus, marginal
correlation operators \eqref{gBigfromDFB} are cumulants (semi-invariants) of the marginal density operators.

As follows from formula \eqref{averagegs} and its generalization \cite{GP} the
marginal density operators $F_s(t)$ are defined in terms of the correlation operators of clusters of
particles $g^{(s)}(t)=(g_{1+0}(t,\{Y\}),\ldots,$ $g_{1+n}(t,\{Y\},s+1,\ldots,s+n),\ldots)$ by the expansion
\begin{eqnarray}\label{FClusters}
    &&F_{s}(t,Y)\doteq\sum\limits_{n=0}^{\infty}\frac{1}{n!}\,
       \mathrm{Tr}_{s+1,\ldots,s+n}\,\,g_{1+n}(t,\{Y\},s+1,\ldots,s+n),\quad s\geq1,
\end{eqnarray}
where the sequence $g_{1+n}(t,\{Y\},s+1,\ldots,s+n),\,n\geq0$, is a solution of the Cauchy problem
of the von Neumann hierarchy for correlation operators of particle clusters \cite{GP}, namely
\begin{eqnarray}\label{rozvNF-N_F_clusters}
   &&g_{1+n}(t,\{Y\},X\setminus Y)=\mathcal{G}(t;\{Y\},X\setminus Y|g(0))\doteq\\
   &&\doteq\sum\limits_{\mathrm{P}:\,(\{Y\},\,X\setminus Y)=\bigcup_i X_i}
      \mathfrak{A}_{|\mathrm{P}|}\big(t,\{\theta(X_1)\},\ldots,\{\theta(X_{|\mathrm{P}|})\}\big)
      \prod_{X_i\subset \mathrm{P}}g_{|X_i|}(0,X_i),\quad s\geq 1,\,n\geq 0,\nonumber
\end{eqnarray}
where $\mathfrak{A}_{|\mathrm{P}|}(t)$ is the $|\mathrm{P}|th$-order cumulant defined by formula
(\ref{cumulantP}). According to estimate (\ref{gEstimate}), series (\ref{FClusters}) exists
and the estimate holds:
$\big\|F_s(t)\big\|_{\mathfrak{L}^{1}(\mathcal{H}_{s})}\leq e^3\emph{c}\sum_{n=0}^{\infty}e^{3n}\emph{c}^n$.
We note that every term of marginal correlation operator expansion \eqref{Gexpg}
is determined by the $(s+n)$-particle correlation operator \eqref{rozvNh} as contrasted to marginal
density operator expansion \eqref{FClusters} which is defined by the $(1+n)$-particle correlation
operator \eqref{rozvNF-N_F_clusters}.

The correlation operators of particle clusters $g^{(s)}(t)=(g_{1+0}(t,\{Y\}),\ldots,g_{1+n}(t,\{Y\},
X\setminus Y),$ $\ldots)\in\mathfrak{L}^{1}(\oplus_{n=0}^{\infty}\mathcal{H}_{s+n})$ can be expressed
in terms of correlation operators of particles (\ref{rozvNh})
\begin{eqnarray}\label{gCluster}
  &&g_{1+n}(t,\{Y\},X\setminus Y)=\\
  &&=\sum\limits_{\mathrm{P}:(\{Y\},\,X\setminus Y)=\bigcup_i X_i}
     (-1)^{|\mathrm{P}|-1}(|\mathrm{P}| -1)!\,
     \prod_{X_i\subset \mathrm{P}}\,\,\sum\limits_{\mathrm{P'}:\,\theta(X_{i})=
     \bigcup_{j_i} Z_{j_i}}\hskip1mm\prod_{Z_{j_i}\subset \mathrm{P'}}g_{|Z_{j_i}|}(t,Z_{j_i}).\nonumber
\end{eqnarray}
In particular case $n=0$, i.e. the correlation operator of a cluster of $|Y|$ particles, these
relations take the form
\begin{eqnarray*}\label{gCluster0}
  &&g_{1+0}(t,\{Y\})=\sum\limits_{\mathrm{P}:\,Y=\bigcup_{i} X_{i}}\hskip1mm
      \prod_{X_{i}\subset \mathrm{P}}g_{|X_{i}|}(t,X_{i}).
\end{eqnarray*}
By the way we observe on that cluster expansions \eqref{FG} follow from  definitions \eqref{Gexpg}
and \eqref{FClusters} in consequence of relations \eqref{gCluster} between correlation operators of
particle clusters and correlation operators of particles.

The marginal ($s$-particle) density operators (\ref{FClusters}) are determined by the Cauchy problem
of the quantum BBGKY hierarchy \cite{BogLect}
\begin{eqnarray}
  \label{BBGKY}
   &&\frac{d}{dt}F_{s}(t,Y)=-\mathcal{N}_{s}(Y)F_{s}(t,Y)+
     \sum\limits_{i\in Y}\mathrm{Tr}_{s+1}(-\mathcal{N}_{\mathrm{int}}(i,s+1))F_{s+1}(t),\\
     \nonumber\\
  \label{BBGKYi}
   &&F_{s}(t)\mid_{t=0}=F_{s}(0),\quad s\geq 1.
\end{eqnarray}
We remind that usually the marginal density operators $F_s(t),\,s\geq1$, are defined by the
well-known formula of the nonequilibrium grand canonical ensemble \cite{GP83,CGP97} in terms
of the density operators $D=(I,D_{1}(t),\ldots,D_{n}(t),\ldots)$ governed by the von Neumann
equations (the quantum Liouville equation)
\begin{eqnarray*}\label{F(D)}
  &&F_{s}(t,Y)=(I,D(t))^{-1}\sum\limits_{n=0}^{\infty}\frac{1}{n!}
     \mathrm{Tr}_{s+1,\ldots,s+n}D_{s+n}(t,X),
\end{eqnarray*}
where $(I,D(t))={\sum\limits}_{n=0}^{\infty}\frac{1}{n!}\mathrm{Tr}_{1,\ldots,n}D_{n}(t)$ is a
normalizing factor, $I$ is the identity operator and $Y\equiv(1,\ldots,s)$, $X\equiv(1,\ldots,s+n)$.
Thus, along with the definition within the framework of the non-equilibrium grand canonical ensemble
the marginal density operators can be defined within the framework of dynamics of correlations that
allows to give the rigorous meaning of the states for more general classes of operators than the
trace-class operators.

If $F(0)\in\mathfrak{L}_{\alpha}^{1}(\mathcal{F}_\mathcal{H})=
{\bigoplus}_{n=0}^{\infty}{\alpha}^{n}\mathfrak{L}_{\alpha}^{1}(\mathcal{H}_n)$ and $\alpha>e$,
then for $t\in\mathbb{R}$ a unique solution of the Cauchy problem \eqref{BBGKY}-\eqref{BBGKYi}
of the quantum BBGKY hierarchy exists and is given by the expansion \cite{GerShJ},\cite{DP}
\begin{eqnarray}\label{RozvBBGKY}
   &&F_{s}(t,Y)=\sum\limits_{n=0}^{\infty}\frac{1}{n!}\,\mathrm{Tr}_{s+1,\ldots,{s+n}}\,
       \mathfrak{A}_{1+n}(t,\{Y\},\, X\backslash Y)F_{s+n}(0,X), \quad s\geq1,
\end{eqnarray}
where the $(1+n)th$-order cumulant $\mathfrak{A}_{1+n}(t)$ of groups of operators
\eqref{groupG} is defined by
\begin{eqnarray}\label{cumulant1+n}
   &&\hskip-5mm\mathfrak{A}_{1+n}(t,\{Y\},\,X\backslash Y)=
     \sum\limits_{\mathrm{P}\,:(\{Y\},\,X\setminus Y)=
     {\bigcup\limits}_i X_i}(-1)^{|\mathrm{P}|-1}(|\mathrm{P}|-1)!
     \prod_{X_i\subset\mathrm{P}}\mathcal{G}_{|\theta(X_i)|}(-t,\theta(X_i)),
\end{eqnarray}
${\sum\limits}_\mathrm{P}$ is the sum over all possible partitions $\mathrm{P}$ of the set
$(\{Y\},\,X\setminus Y)$ into $|\mathrm{P}|$ nonempty mutually disjoint subsets
$X_i\subset(\{Y\},\,X\setminus Y)$.

Formally, the evolution equations for marginal correlation operators are derived from the quantum
BBGKY hierarchy for marginal density operators (\ref{BBGKY}) on basis of expression (\ref{gBigfromDFB}).
Then the evolution of all possible states of quantum many-particle systems obeying the Maxwell-Boltzmann
statistics with the Hamiltonian (\ref{H_n}) can be described within the framework of marginal correlation
operators governed by the nonlinear quantum BBGKY hierarchy
\begin{eqnarray}
 \label{gBigfromDFBa}
   &&\frac{d}{dt}G_s(t,Y)=\mathcal{N}(Y\mid G(t))+
      \mathrm{Tr}_{s+1}\sum_{i\in Y}(-\mathcal{N}_{\mathrm{int}}(i,s+1))\big(G_{s+1}(t,Y,s+1)+\\
   &&+\sum_{\mbox{\scriptsize$\begin{array}{c}\mathrm{P}:(Y,s+1)=X_1\bigcup X_2,\\i\in
      X_1;s+1\in X_2\end{array}$}}G_{|X_1|}(t,X_1)G_{|X_2|}(t,X_2)\big),\nonumber\\ \nonumber\\
 \label{gBigfromDFBai}
   &&G_{s}(t,Y)\big|_{t=0}=G_{s}(0,Y),\quad s\geq1.
\end{eqnarray}
In equation \eqref{gBigfromDFBa} the operator $\mathcal{N}(Y\mid G(t))$ is generator of the von Neumann
hierarchy (\ref{vNh}) defined by formula \eqref{vNgenerator}, i.e.
\begin{eqnarray*}\label{Nnl}
   &&\mathcal{N}(Y\mid G(t))\doteq(-\mathcal{N}_{s}(Y))G_s(t,Y)+\\
   &&+\sum\limits_{\mathrm{P}:\,Y=X_{1}\bigcup X_2}\,\sum\limits_{i_{1}\in X_{1}}
      \sum\limits_{i_{2}\in X_{2}}
      (-\mathcal{N}_{\mathrm{int}}(i_{1},i_{2}))G_{|X_{1}|}(t,X_{1})G_{|X_{2}|}(t,X_{2}),\nonumber
\end{eqnarray*}
where the operators $(-\mathcal{N}_{s})$ and $(-\mathcal{N}_{\mathrm{int}})$ are defined by
(\ref{komyt}) and (\ref{oper Nint2}) respectively, ${\sum\limits}_{\mathrm{P}:\,
Y=X_{1}\bigcup X_2}$ is the sum over all possible partitions $\mathrm{P}$ of the set
$Y\equiv(1,\ldots,s)$ into two nonempty mutually disjoint subsets $X_1\subset Y$ and $X_2\subset Y$, and
$\sum_{\mbox{\scriptsize$\begin{array}{c}\mathrm{P}:(Y,s+1)=X_1\bigcup X_2,\\i\in X_1;s+1\in X_2\end{array}$}}$
is the sum over all possible partitions of the set $(Y,s+1)$ into two mutually disjoint subsets $X_1$ and
$X_2$ such that $ith$ particle belongs to the subset $X_1$ and $(s+1)th$ particle belongs to $X_2$.
As far as we know hierarchy (\ref{gBigfromDFBa}) was introduced by M.M. Bogolyubov \cite{BogLect} and
in the papers of J. Yvon \cite{Y} and M.S. Green \cite{Gre56} for systems of classical particles.

Another method of the justification of evolution equations for marginal correlation operators consists
in their derivation from the von Neumann hierarchy for correlation operators (\ref{vNh}) on the basis of
definition (\ref{Gexpg}).

\subsection{The derivation of the nonlinear quantum BBGKY hierarchy}
In this section we establish that marginal correlation operators (\ref{Gexpg}) are governed by the
nonlinear quantum BBGKY hierarchy (\ref{gBigfromDFBa}). With this aim we differentiate by time variable
the marginal correlation operators defined by series (\ref{Gexpg}) in the sense of the pointwise
convergence on the space $\mathfrak{L}^{1}(\mathcal{H}_{s})$
\begin{eqnarray}\label{Gexpg_2}
   &&\frac{d}{dt}G_{s}(t,Y)=\sum\limits_{n=0}^{\infty}\frac{1}{n!}\,
      \mathrm{Tr}_{s+1,\ldots,s+n}\,\,\big((-\mathcal{N}_{s+n}(X))g_{s+n}(t,X)+\nonumber\\
   &&+\sum_{\mathrm{P}:X=X_1\cup X_2}\sum_{i_1\in X_1}
      \sum_{i_2\in X_2}(-\mathcal{N}_{\mathrm{int}}(i_{1},i_{2}))g_{|X_1|}(t,X_1)g_{|X_2|}(t,X_2)\big),
\end{eqnarray}
where we use the notations: $X\equiv(1,\ldots,s+n)$, $Y\equiv(1,\ldots,s)$, ${\sum\limits}_{\mathrm{P}:\,
Y=X_{1}\bigcup X_2}$ is the sum over all possible partitions $\mathrm{P}$ of the set
$Y\equiv(1,\ldots,s)$ into two nonempty mutually disjoint subsets $X_1\subset Y$ and $X_2\subset Y$,
and operators $(-\mathcal{N}_{s+n})$ and $(-\mathcal{N}_{\mathrm{int}})$ are defined by formulas
(\ref{komyt}) and (\ref{oper Nint2}) respectively. Taking into account the equality
\begin{eqnarray*}
   &&-\mathcal{N}_{s+n}(X)=-\mathcal{N}_{s}(Y)-\mathcal{N}_{n}(X\backslash Y)+
     \sum_{i_1\in Y}\sum_{i_2\in X\backslash Y}(-\mathcal{N}_{\mathrm{int}}(i_{1},i_{2})),
\end{eqnarray*}
and the identity
\begin{eqnarray*}
   &&\mathrm{Tr}_{s+1,\ldots,s+n}(-\mathcal{N}_{n}(X\setminus Y))g_{s+n}(t)=0,
\end{eqnarray*}
and according to the symmetry property of the correlation operators $g_{s+n}(t,X)$, for the first
term in the right-hand side of identity (\ref{Gexpg_2}) in terms of definition (\ref{Gexpg}) we obtain
\begin{eqnarray}\label{d1}
   &&\sum\limits_{n=0}^{\infty}\frac{1}{n!}\,
      \mathrm{Tr}_{s+1,\ldots,s+n}\,\,(-\mathcal{N}_{s+n}(X))g_{s+n}(t,1,\ldots,s+n)= \\
   &&=-\mathcal{N}_s(Y)G_s(t,Y)+\sum\limits_{n=1}^{\infty}\frac{1}{n!}\,\mathrm{Tr}_{s+1,\ldots,s+n}\,
      \sum_{i_1\in Y}\sum_{i_2\in X \backslash Y}(-\mathcal{N}_{\mathrm{int}}(i_{1},i_{2}))g_{s+n}(t,X)=\nonumber\\
   &&=-\mathcal{N}_s(Y)G_s(t,Y)+\sum\limits_{n=0}^{\infty}\frac{1}{n!}\,
      \mathrm{Tr}_{s+1,\ldots,s+n+1}
      \sum_{i_1\in Y}(-\mathcal{N}_{\mathrm{int}}(i_{1},s+1))g_{s+n+1}(t,X,s+n+1)=\nonumber
\end{eqnarray}
\begin{eqnarray*}
   &&=-\mathcal{N}_s(Y)G_s(t,Y)+\mathrm{Tr}_{s+1}\sum_{i_1\in Y}(-\mathcal{N}_{\mathrm{int}}(i_{1},s+1))
      G_{s+1}(t,1,\ldots,s+1).\nonumber
\end{eqnarray*}
Let us consider successively the following four parts of the second term of the right-hand side
of identity (\ref{Gexpg_2})
\begin{eqnarray}\label{sums}
  &&\sum_{\mathrm{P}:X=X_1\cup X_2}\sum_{i_1\in X_1}
    \sum_{i_2\in X_2}(-\mathcal{N}_{\mathrm{int}}(i_{1},i_{2}))=\sum_{\mathrm{P}:X=X_1\cup X_2}
    \big(\sum_{i_1\in Y\cap X_1}\sum_{i_2\in Y\cap X_2}+\\
  &&+\sum_{i_1\in X\backslash Y\cap X_1}\sum_{i_2\in Y\cap X_2}+
    \sum_{i_1\in Y\cap X_1}\sum_{i_2\in X\backslash Y\cap X_2}+\sum_{i_1\in X\backslash Y\cap X_1}
    \sum_{i_2\in X\backslash Y\cap X_2}\big)(-\mathcal{N}_{\mathrm{int}}(i_{1},i_{2})).\nonumber
\end{eqnarray}
Taking into account that the equality is true
\begin{eqnarray*}
   &&\sum_{\mbox{\scriptsize$\begin{array}{c}\mathrm{P}:X=X_1\cup X_2,\\
      Y\cap X_1\neq\emptyset,Y\cap X_2\neq\emptyset\end{array}$}}g_{|X_1|}(t,X_1)g_{|X_2|}(t,X_2)=\\
   &&=\sum_{\mbox{\scriptsize$\begin{array}{c}\mathrm{P}:Y=Y_{1}\bigcup Y_2,\\
      Y_{1}\neq\emptyset,Y_{2}\neq\emptyset\end{array}$}}\sum\limits_{Z\subset X\setminus Y}
      g_{|Y_{1}|+|Z|}(t,Y_{1},Z)g_{|Y_{2}|+|(X\setminus Y)\setminus Z|}(t,Y_{2},(X\setminus Y)\setminus Z),
\end{eqnarray*}
and the validity of the following equality (according to the symmetry property of operators $g_n(t)$)
\begin{eqnarray*}
   &&\mathrm{Tr}_{s+1,\ldots,s+n}\sum\limits_{Z\subset X\setminus Y}\,
      g_{|Y_{1}|+|Z|}(t,Y_{1},Z)g_{|Y_{2}|+|(X\setminus Y)\setminus Z|}(t,Y_{2},(X\setminus Y)\setminus Z)=\\
   &&=\mathrm{Tr}_{s+1,\ldots,s+n}\sum\limits_{k=0}^{n}\frac{n!}{k!(n-k)!}
      g_{|Y_{1}|+n-k}(t,Y_{1},s+1,\ldots,s+n-k)\times\\
   &&\times g_{|Y_{2}|+k}(t,Y_{2},s+n-k+1,\ldots,s+n),
\end{eqnarray*}
for the first part of equality \eqref{sums} of the second term of identity \eqref{Gexpg_2} it holds
\begin{eqnarray*}
   &&\mathrm{Tr}_{s+1,\ldots,s+n}\sum_{\mathrm{P}:X=X_1\cup X_2}\sum_{i_1\in Y\cap X_1}
      \sum_{i_2\in Y\cap X_2}(-\mathcal{N}_{\mathrm{int}}(i_{1},i_{2}))g_{|X_1|}(t,X_1)g_{|X_2|}(t,X_2)=\\
   &&=\mathrm{Tr}_{s+1,\ldots,s+n}\sum\limits_{\mathrm{P}:\,Y=Y_{1}\bigcup Y_2}\,
      \sum\limits_{i_{1}\in Y_{1}}\sum\limits_{i_{2}\in Y_{2}}
      (-\mathcal{N}_{\mathrm{int}}(i_{1},i_{2}))\sum\limits_{k=0}^{n}\frac{n!}{k!(n-k)!}\times\\
   &&\times g_{|Y_{1}|+n-k}(t,Y_{1},s+1,\ldots,s+n-k)g_{|Y_{2}|+k}(t,Y_{2},s+n-k+1,\ldots,s+n).
\end{eqnarray*}
Then in terms of definition (\ref{Gexpg}) the last expression takes the form
\begin{eqnarray*}
    &&\sum\limits_{n=0}^{\infty}\frac{1}{n!}\mathrm{Tr}_{s+1,\ldots,s+n}
      \sum\limits_{\mathrm{P}:\,Y=Y_{1}\bigcup Y_2}\,\sum\limits_{i_{1}\in Y_{1}}\sum\limits_{i_{2}\in Y_{2}}
      (-\mathcal{N}_{\mathrm{int}}(i_{1},i_{2}))\sum\limits_{k=0}^{n}\frac{n!}{k!(n-k)!}\times\\
    &&\times g_{|Y_{1}|+n-k}(t,Y_{1},s+1,\ldots,s+n-k)g_{|Y_{2}|+k}(t,Y_{2},s+n-k+1,\ldots,s+n)=\\
    &&=\sum\limits_{\mathrm{P}:\,Y=Y_{1}\bigcup Y_2}\,\sum\limits_{i_{1}\in Y_{1}}\sum\limits_{i_{2}\in Y_{2}}
      (-\mathcal{N}_{\mathrm{int}}(i_{1},i_{2}))G_{|Y_{1}|}(t,Y_{1})G_{|Y_{2}|}(t,Y_{2}).
\end{eqnarray*}
Hence for the first part of equality \eqref{sums} of the second term of the right-hand side of identity 
(\ref{Gexpg_2}) we have
\begin{eqnarray}\label{d3}
   &&\hskip-7mm\sum\limits_{n=0}^{\infty}\frac{1}{n!}\,
     \mathrm{Tr}_{s+1,\ldots,s+n}\,\sum_{\mathrm{P}:X=X_1\cup X_2}\sum_{i_1\in Y\cap X_1}
     \sum_{i_2\in Y\cap X_2}(-\mathcal{N}_{\mathrm{int}}(i_{1},i_{2}))g_{|X_1|}(t,X_1)g_{|X_2|}(t,X_2)=\\
   &&\hskip-7mm=\sum\limits_{\mathrm{P}:\,Y=Y_{1}\bigcup Y_2}\,
     \sum\limits_{i_{1}\in Y_{1}}\sum\limits_{i_{2}\in Y_{2}}
     (-\mathcal{N}_{\mathrm{int}}(i_{1},i_{2}))G_{|Y_{1}|}(t,Y_{1})G_{|Y_{2}|}(t,Y_{2}).\nonumber
\end{eqnarray}
Similarly the second and third parts of equality \eqref{sums} of the second term of the right-hand side
of identity (\ref{Gexpg_2}) are expressed in terms of definition (\ref{Gexpg}) in the form
\begin{eqnarray}\label{d4}
   &&\sum\limits_{n=0}^{\infty}\frac{1}{n!}\,
     \mathrm{Tr}_{s+1,\ldots,s+n}\,\sum_{\mathrm{P}:X=X_1\cup X_2}
     \big(\sum_{i_1\in X\backslash Y\cap X_1}\sum_{i_2\in Y\cap X_2}+\\
   &&+\sum_{i_1\in Y\cap X_1}\sum_{i_2\in X\backslash Y\cap X_2}\big)
     (-\mathcal{N}_{\mathrm{int}}(i_{1},i_{2}))g_{|X_1|}(t,X_1)g_{|X_2|}(t,X_2)=\nonumber\\
   &&\hskip-5mm=\sum_{i\in Y}\mathrm{Tr}_{s+1}(-\mathcal{N}_{\mathrm{int}}(i,s+1))
     \sum_{\mbox{\scriptsize $\begin{array}{c}\mathrm{P}:(Y,s+1)=X_1\bigcup X_2,\\i\in X_1;s+1\in X_2 \end{array}$}}G_{|X_1|}(t,X_1)G_{|X_2|}(t,X_2),\nonumber
\end{eqnarray}
where $\sum_{\mbox{\scriptsize$\begin{array}{c}\mathrm{P}:(Y,s+1)=X_1\bigcup X_2,\\i\in X_1;s+1\in X_2\end{array}$}}$
is the sum over all possible partitions of the set $(Y,s+1)$ into two mutually disjoint subsets $X_1$ and $X_2$ such
that $ith$ particle belongs to the subset $X_1$ and $(s+1)th$ particle belongs to $X_2$.

Then taking into account the following identity for the fourth part \eqref{sums} of the second term of (\ref{Gexpg_2})
\begin{eqnarray*}
   &&\hskip-10mm\sum\limits_{n=0}^{\infty}\frac{1}{n!}\,
      \mathrm{Tr}_{s+1,\ldots,s+n}\,\sum_{\mathrm{P}:X=X_1\cup X_2}\sum_{i_1\in X\backslash Y\cap X_1}
      \sum_{i_2\in X\backslash Y\cap X_2}(-\mathcal{N}_{\mathrm{int}}(i_{1},i_{2}))g_{|X_1|}(t,X_1)g_{|X_2|}(t,X_2)=0,
\end{eqnarray*}
and identities \eqref{d3},\eqref{d4}, for the second term of the right-hand side of (\ref{Gexpg_2}) it holds
\begin{eqnarray}\label{d2}
   &&\sum\limits_{n=0}^{\infty}\frac{1}{n!}\,
     \mathrm{Tr}_{s+1,\ldots,s+n}\,\sum_{\mathrm{P}:X=X_1\cup X_2}\sum_{i_1\in X_1}
      \sum_{i_2\in X_2}(-\mathcal{N}_{\mathrm{int}}(i_{1},i_{2}))g_{|X_1|}(t,X_1)g_{|X_2|}(t,X_2)=\\
   &&=\sum\limits_{\mathrm{P}:\,Y=Y_{1}\bigcup Y_2}\,\sum\limits_{i_{1}\in Y_{1}}\sum\limits_{i_{2}\in Y_{2}}
      (-\mathcal{N}_{\mathrm{int}}(i_{1},i_{2}))G_{|Y_{1}|}(t,Y_{1})G_{|Y_{2}|}(t,Y_{2})+\nonumber\\
   &&+\sum_{i\in Y}\mathrm{Tr}_{s+1}(-\mathcal{N}_{\mathrm{int}}(i,s+1))
      \sum_{\mbox{\scriptsize $\begin{array}{c}\mathrm{P}:(Y,s+1)=
      X_1\bigcup X_2,\\i\in X_1;s+1\in X_2 \end{array}$}}G_{|X_1|}(t,X_1)G_{|X_2|}(t,X_2).\nonumber
\end{eqnarray}

In consequence of identities \eqref{d1} and \eqref{d2} we finally derive
\begin{eqnarray*}
   &&\hskip-8mm\frac{d}{dt}G_{s}(t,Y)=-\mathcal{N}_s(Y)G_s(t,Y)+\sum\limits_{\mathrm{P}:\,Y=Y_{1}\bigcup Y_2}\,
     \sum\limits_{i_{1}\in Y_{1}}\sum\limits_{i_{2}\in Y_{2}}
     (-\mathcal{N}_{\mathrm{int}}(i_{1},i_{2})G_{|Y_{1}|}(t,Y_{1})G_{|Y_{2}|}(t,Y_{2})+
\end{eqnarray*}
\begin{eqnarray*}
   &&\hskip-8mm+\sum_{i\in Y}\mathrm{Tr}_{s+1}(-\mathcal{N}_{\mathrm{int}}(i,s+1))(G_{s+1}(t)
     +\sum_{\mbox{\scriptsize $\begin{array}{c}\mathrm{P}:(Y,s+1)=
     X_1\bigcup X_2,\\i\in X_1;s+1\in X_2 \end{array}$}}G_{|X_1|}(t,X_1)G_{|X_2|}(t,X_2)),\nonumber
\end{eqnarray*}
where we use notations accepted above in \eqref{Gexpg_2},\eqref{d4}. The constructed identity for
the marginal correlation operators defined by expansion (\ref{Gexpg}) we treat as the hierarchy of
evolution equations, which governs the marginal correlation operators of quantum many-particle systems.

We also formulate the nonlinear quantum BBGKY hierarchy in case of many-particle systems obeying
quantum statistics with the Hamiltonian
\begin{eqnarray*}\label{Hn}
    &&H_{n}=\sum\limits_{i=1}^{n}K(i)+\sum\limits_{k=1}^{n}\sum\limits_{i_{1}<\ldots<i_{k}=1}^{n}
        \Phi^{(k)}(q_{i_{1}},\ldots,q_{i_{k}}),
\end{eqnarray*}
where $\Phi^{(k)}$ is a $k$-body interaction potential. We introduce the symmetrization operator
$\mathcal{S}^{+}_{n}$ and the anti-symmetrization operator $\mathcal{S}^{-}_{n}$ on
$\mathcal{H}^{\otimes n}$ which are defined on the space $\mathcal{H}_{n}$ by the formula
\begin{eqnarray*}\label{Sn}
    &&\mathcal{S}^{\pm}_{n}\doteq\frac {1}{n!}
       \sum\limits_{\pi\epsilon \mathfrak{S}_{n}}(\pm1)^{|\pi|}p_{\pi},
\end{eqnarray*}
where the operator $p_{\pi}$ is a transposition operator of the permutation $\pi$ from the permutation
group $\mathfrak{S}_{n}$ of the set $(1,\ldots,n)$ and $|\pi|$ denotes the number of transpositions in
the permutation $\pi$. The operators $\mathcal{S}^{\pm}_{n}$ are orthogonal projectors, i.e.
$({\mathcal{S}^{\pm}_{n}})^{2}=\mathcal{S}^{\pm}_{n}$, ranges of which are correspondingly the symmetric
tensor product $\mathcal{H}_{n}^{+}$ and the antisymmetric tensor product $\mathcal{H}_{n}^{-}$ of $n$
Hilbert spaces $\mathcal{H}$. We denote by $\mathcal{F}^{+}_{\mathcal{H}}=
{\bigoplus\limits}_{n=0}^{\infty}\mathcal{H}_{n}^{+}$ and $\mathcal{F}^{-}_{\mathcal{H}}=
{\bigoplus\limits}_{n=0}^{\infty}\mathcal{H}_{n}^{-}$ the Bose and Fermi Fock spaces over the Hilbert
space $\mathcal{H}$ respectively.

The evolution of all possible states of many-particle systems of bosons or fermions is described
within the framework of marginal correlation operators
$G(t)=(G_0,G_{1}(t),\ldots,G_{s}(t),\ldots)\in\mathfrak{L}^{1}(\mathcal{F}_\mathcal{H}^{\pm})$ governed
by the following nonlinear quantum BBGKY hierarchy
\begin{eqnarray}\label{eqG}
  &&\hskip-5mm\frac{d}{dt}G_s(t,Y)=-\mathcal{N}_{s}(Y)G_{s}(t,Y)+\\
  &&\hskip-5mm+\sum\limits_{\mbox{\scriptsize $\begin{array}{c}\mathrm{P}:Y=\bigcup_{i}X_{i},\\
     |\mathrm{P}|>1\end{array}$}}\hskip-2mm\sum\limits_{\mbox{\scriptsize$\begin{array}{c}{Z_{1}\subset X_{1}},\\
     Z_{1}\neq\emptyset\end{array}$}}\hskip-2mm\ldots\hskip-2mm
     \sum\limits_{\mbox{\scriptsize $\begin{array}{c} {Z_{|\mathrm{P}|}\subset X_{|\mathrm{P}|}},\\
     Z_{|\mathrm{P}|}\neq\emptyset \end{array}$}}
     \hskip-2mm \big(-\mathcal{N}_{\mathrm{int}}^{(\sum\limits_{r=1}^{|\mathrm{P}|}|Z_{r}|)}
     (Z_{1},\ldots,Z_{|\mathrm{P}|})\big)\mathcal{S}^{\pm}_{s}
     \prod_{X_{i}\subset\mathrm{P}}G_{|X_{i}|}(t,X_{i})+\nonumber\\
  &&\hskip-5mm+\sum_{n=1}^\infty\sum_{k=1}^n\sum_{1=j_1<\ldots<j_k}^s \mathrm{Tr}_{s+1,\ldots,s+1+n-k}
     (-\mathcal{N}_{\mathrm{int}}^{(n+1)}(j_1,\ldots,j_k,s+1,\ldots,s+1+n-k))\times\nonumber\\
  &&\hskip-5mm\times\sum_{\mbox{\scriptsize$\begin{array}{c}\mathrm{P}:(1,\ldots,s+1+n-k)=\bigcup_{i=1}X_i,
     \\X_i\not \subseteq Y\backslash(j_1,\ldots,j_k),\,|\mathrm{P}|\leq n+1\end{array}$}}
     \mathcal{S}^{\pm}_{s+1+n-k}\prod_{X_{i}\subset\mathrm{P}}G_{|X_i|}(t,X_i),\nonumber
\end{eqnarray}
where notations accepted above are used and the operator $(-\mathcal{N}_{\mathrm{int}}^{(k)})$ acts
on $\mathfrak{L}^{1}_{0}(\mathcal{H}_{s})\subset\mathfrak{L}^{1}(\mathcal{H}_{s})$ according to the formula
\begin{eqnarray}\label{Nintn}
   &&(-\mathcal{N}_{\mathrm{int}}^{(k)}(i_{1},\ldots,i_{k}))f_{s}\doteq
      -\frac{i}{\hbar}(\Phi^{(k)}(i_{1},\ldots,i_{k})f_{s}-f_{s}\Phi^{(k)}(i_{1},\ldots,i_{k})).
\end{eqnarray}

We emphasize that the evolution of marginal correlation operators of both finitely and infinitely many
quantum particles is described by initial-value problem of the nonlinear quantum BBGKY hierarchy
\eqref{gBigfromDFBa} (or \eqref{eqG}). For finitely many particles the nonlinear quantum BBGKY hierarchy
is equivalent to the von Neumann hierarchy \eqref{vNh}.

%%%%%%%%%%%%%%%%%%%%%%%%%%%%%%%%%%%%%%%%%%%%%%%%%%%%%%%%%%%%%%%%%%%%%%%%%%%%%%%%%%%%%%%%%%%%%%%%%%%%%%%%%%%%%%%%%%%%%%%%%%
%%%%%%%%%%%%%%%%%%%%%%%%%%%%%%%%%%%%%%%%%%%%%%%%%%%%%%%%%%%%%%%%%%%%%%%%%%%%%%%%%%%%%%%%%%%%%%%%%%%%%%%%%%%%%%%%%%%%%%%%%%

\section{A nonperturbative solution of the nonlinear quantum BBGKY hierarchy}
In this section we construct a solution of the initial-value problem of the nonlinear quantum BBGKY
hierarchy (\ref{gBigfromDFBa}). A nonperturbative solution is represented in form of an expansion
over particle clusters which evolution is governed by the corresponding-order cumulant (semi-invariant)
of nonlinear groups of operators (\ref{rozvNh}) generated by the von Neumann hierarchy (\ref{vNh}).
A solution representation in the form of the perturbation (iteration) series of hierarchy (\ref{gBigfromDFBa})
is derived as a result of applying of analogs of the Duhamel equation to cumulants of groups of operators
(\ref{rozvNh}) of constructed solution.

\subsection{A solution in case of chaos initial data}
To construct a nonperturbative solution of the Cauchy problem (\ref{gBigfromDFBa})-(\ref{gBigfromDFBai})
of the nonlinear quantum BBGKY hierarchy we first consider its structure for physically motivated example
of initial data, namely, initial data satisfying a chaos property
\begin{eqnarray}\label{inG}
   &&G_s(t,Y)|_{t=0}=G_1(0,1)\delta_{s,1},\quad s\geq1,
\end{eqnarray}
where $\delta_{s,1}$ is a Kronecker symbol. Chaos property (\ref{inG}) means the absence of state
correlations in a system at the initial time.

According to definition (\ref{Gexpg}) and solution expansion (\ref{gth}), in the case under
consideration the following relation between the marginal correlation operators and correlation
operators is true
\begin{eqnarray}\label{Gg}
   &&G_1(0,i)=g_1(0,i).
\end{eqnarray}
Taking into account the form (\ref{gth}) of a solution of the initial-value problem of the
von Neumann hierarchy (\ref{vNh}) in case of initial data (\ref{gChaos}), for expansion
(\ref{Gexpg}) we obtain
\begin{eqnarray}\label{Gg(0)}
   &&G_{s}(t,Y)=\sum\limits_{n=0}^{\infty}\frac{1}{n!}
       \,\mathrm{Tr}_{s+1,\ldots, s+n}\,\mathfrak{A}_{s+n}(t,1,\ldots,s+n)\,
       \prod_{i=1}^{s+n}g_{1}(0,i),
\end{eqnarray}
where $\mathfrak{A}_{s+n}(t)$ is $(s+n)th$-order cumulant (\ref{cumulantP}). In consequence
of relation (\ref{Gg}) we finally derive
\begin{eqnarray}\label{GUG(0)}
   &&G_{s}(t,Y)=\sum\limits_{n=0}^{\infty}\frac{1}{n!}
       \,\mathrm{Tr}_{s+1,\ldots, s+n}\,\mathfrak{A}_{s+n}(t,1,\ldots,s+n)\,
       \prod_{i=1}^{s+n}G_{1}(0,i),\quad s\geq 1.
\end{eqnarray}
If $\|G_{1}(0)\|_{\mathfrak{L}^{1}(\mathcal{H})}\leq (2e)^{-1}$, series (\ref{GUG(0)}) converges,
since for cumulants (\ref{cumulantP}) the estimate holds \cite{GerShJ}
\begin{eqnarray*}
   &&\|\mathfrak{A}_{n}(t)f\|_{\mathfrak{L}^{1}(\mathcal{H}_n)}\leq\,n!\,e^{n}\,
      \|f\|_{\mathfrak{L}^{1}(\mathcal{H}_n)}.
\end{eqnarray*}
From the structure of series (\ref{GUG(0)}) it is clear that in case of absence of correlations
at initial instant in a system the correlations generated by the dynamics of quantum many-particle
systems are completely governed by cumulants (\ref{cumulantP}) of groups of operators (\ref{groupG}).

Thus, the cumulant structure of solution (\ref{rozvNh}) of the von Neumann hierarchy (\ref{vNh})
induces the cumulant structure of solution expansion (\ref{GUG(0)}) of the initial-value problem
of the quantum nonlinear BBGKY hierarchy for marginal correlation operators.

The evolution equations which satisfy expression (\ref{GUG(0)}) are derived similarly to the derivation
of hierarchy (\ref{gBigfromDFBa}) given in section 3.2 on the base of definition (\ref{Gg(0)}).

We note, that in case of initial data \eqref{gChaos} solution \eqref{gth} of the Cauchy problem
\eqref{vNh}-\eqref{vNhi} of the von Neumann hierarchy may be rewritten in another representation.
For $n=1$, we have
\begin{eqnarray*}
   &&g_{1}(t,1)=\mathfrak{A}_{1}(t,1)g_{1}(0,1).
 \end{eqnarray*}
Then, within the context of the definition of the first-order cumulant, $\mathfrak{A}_{1}(-t)$,
and the dual group of operators $\mathfrak{A}_{1}(t)$, we express the correlation operators
$g_{s}(t),\,s\geq 2$, in terms of the one-particle correlation operator $g_{1}(t)$ using formula
\eqref{gth}. Hence for $s\geq2$ formula \eqref{gth} is represented in the form of the functional
with respect to one-particle correlation operators
\begin{eqnarray*}
   &&g_{s}(t,Y\mid g_{1}(t))=\widehat{\mathfrak{A}}_{s}(t,Y)\,\prod_{i=1}^{s}\,g_{1}(t,i),\quad s\geq 2,
\end{eqnarray*}
where $\widehat{\mathfrak{A}}_{s}(t,Y)$ is $sth$-order cumulant \eqref{cumulantP} of the scattering
operators
\begin{eqnarray}\label{so}
   &&\widehat{\mathcal{G}}_{s}(t,Y)\doteq
      \mathcal{G}_{s}(-t,Y)\prod_{i=1}^{s}\mathcal{G}_{1}(t,i),\quad s\geq1.
\end{eqnarray}
The generator of the scattering operator $\widehat{\mathcal{G}}_{t}(Y)$ is determined by the operator
\begin{eqnarray*}
  &&\frac{d}{dt}\widehat{\mathcal{G}}_{s}(t,Y)|_{t=0}=\sum\limits_{k=2}^{s}\,\,
     \sum\limits_{i_{1}<\ldots<i_{k}=1}^{s}\,(-\mathcal{N}_{\mathrm{int}}^{(k)}(i_{1},\ldots,i_{k})),
\end{eqnarray*}
where the operator $(-\mathcal{N}_{\mathrm{int}}^{(k)})$ acts on $\mathfrak{L}^{1}_{0}
(\mathcal{H}_{s})\subset\mathfrak{L}^{1}(\mathcal{H}_{s})$ according to formula \eqref{Nintn}.

Similar representation of a solution holds for marginal correlation operators \eqref{GUG(0)} which
forms inherently the basis of the kinetic description of the evolution.
In this case the marginal correlation functionals $G_{s}\big(t,Y\mid G_{1}(t)\big),\,s\geq2$, are
represented by the expansions
\begin{eqnarray}\label{cf}
   &&\hskip-8mmG_{s}\big(t,Y\mid G_{1}(t)\big)=\sum\limits_{n=0}^{\infty}\frac{1}{n!}\,
      \mathrm{Tr}_{s+1,\ldots,s+n}\mathfrak{V}_{1+n}\big(t,\theta(\{Y\}),s+1,\ldots,s+n\big)
      \prod_{i=1}^{s+n}G_{1}(t,i),
\end{eqnarray}
where the operator $G_{1}(t,i)$ is given by \eqref{Gg(0)} for $s=1$, and we use the notion of the
declasterization mapping defined in section 2 \cite{DP}.
In expansion (\ref{cf}) the $(1+n)th$-order evolution operator $\mathfrak{V}_{1+n}(t)$ is defined by the
formula \cite{GT}

\begin{eqnarray*}
    &&\mathfrak{V}_{1+n}(t,\theta(\{Y\},X\setminus Y )\doteq n!\,\sum_{k=0}^{n}\,(-1)^k\,\sum_{n_1=1}^{n}\ldots
        \sum_{n_k=1}^{n-n_1-\ldots-n_{k-1}}\frac{1}{(n-n_1-\ldots-n_k)!}\times\\
    &&\times\widehat{\mathfrak{A}}_{s+n-n_1-\ldots-n_k}(t,\{Y\},s+1,\ldots,s+n-n_1-\ldots-n_k)\times\nonumber\\
    &&\times\prod_{j=1}^k\,\sum\limits_{\mbox{\scriptsize $\begin{array}{c}\mathrm{D}_{j}:Z_j=\bigcup_{l_j} X_{l_j},\\
        |\mathrm{D}_{j}|\leq s+n-n_1-\dots-n_j\end{array}$}}\frac{1}{|\mathrm{D}_{j}|!}
        \sum_{i_1\neq\ldots\neq i_{|\mathrm{D}_{j}|}=1}^{s+n-n_1-\ldots-n_j}\,\,
        \prod_{X_{l_j}\subset \mathrm{D}_{j}}\,\frac{1}{|X_{l_j}|!}\,\,
        \widehat{\mathfrak{A}}_{1+|X_{l_j}|}(t,i_{l_j},X_{l_j}),\nonumber
\end{eqnarray*}
where $\sum_{\mathrm{D}_{j}:Z_j=\bigcup_{l_j} X_{l_j}}$ is the sum over all possible dissections
$\mathrm{D}_{j}$ of the linearly ordered set $Z_j\equiv(s+n-n_1-\ldots-n_j+1,\ldots,s+n-n_1-\ldots-n_{j-1})$
on no more than $s+n-n_1-\ldots-n_j$ linearly ordered subsets, and the operator $\widehat{\mathfrak{A}}_{1+n}(t)$
is the $(1+n)$-order cumulant \eqref{cumulantP} of the scattering operators \eqref{so}. For example, the lower
orders evolution operators $\mathfrak{V}_{1+n}\big(t,\theta(\{Y\}),s+1,\ldots,s+n\big),\,n\geq0$, have the form
\begin{eqnarray*}
   &&\mathfrak{V}_{1}(t,\theta(\{Y\}))=\widehat{\mathfrak{A}}_{s}(t,\theta(\{Y\}),\\
   &&\mathfrak{V}_{2}(t,\theta(\{Y\}),s+1)=\widehat{\mathfrak{A}}_{s+1}(t,\theta(\{Y\}),s+1)-
       \widehat{\mathfrak{A}}_{s}(t,\theta(\{Y\}))\sum_{i=1}^s \widehat{\mathfrak{A}}_{2}(t,i,s+1),\nonumber
\end{eqnarray*}
and in case of $s=2$, it holds
\begin{eqnarray*}
    &&\mathfrak{V}_{1}(t,\theta(\{1,2\}))=\widehat{\mathcal{G}}_{2}(t,1,2)-I.
\end{eqnarray*}

We point out also that in case of chaos initial data solution expansion (\ref{RozvBBGKY}) of the quantum
BBGKY hierarchy (\ref{BBGKY}) for marginal density operators differs from solution expansion (\ref{GUG(0)})
of the nonlinear quantum BBGKY hierarchy (\ref{gBigfromDFBa}) for marginal correlation operators only by
the order of the cumulants of the groups of operators of the von Neumann equations \cite{GP},\cite{DP}
\begin{eqnarray}\label{FUg}
   &&F_{s}(t,Y)=\sum\limits_{n=0}^{\infty}\frac{1}{n!}
       \,\mathrm{Tr}_{s+1,\ldots, s+n}\,\mathfrak{A}_{1+n}(t,\{Y\},X\setminus Y)
       \,\prod_{i=1}^{s+n}F_{1}(0,i), \quad s\geq 1,
\end{eqnarray}
where $\mathfrak{A}_{1+n}(t)$ is the $(1+n)th$-order cumulant (\ref{cumulant1+n}). Series (\ref{FUg})
converges under the condition: $\|F_{1}(0)\|_{\mathfrak{L}^{1}(\mathcal{H})}\leq e^{-1}$.

\subsection{The structure of a nonperturbative solution expansion}
The direct method of the construction of a solution of the nonlinear quantum BBGKY hierarchy
\eqref{gBigfromDFBa} in the form of nonperturbative expansion consists in its derivation on
the basis of expansions \eqref{gBigfromDFB} from nonperturbative solution \eqref{RozvBBGKY} of
initial-value problem of the quantum BBGKY hierarchy \eqref{BBGKY}-\eqref{BBGKYi}. Following stated
above approach, we derive a formula for a solution of the quantum nonlinear BBGKY hierarchy for
marginal correlation operators in case of general initial data on the basis of definition \eqref{Gexpg}
and nonperturbative solution \eqref{rozvNh} of initial-value problem of the von Neumann hierarchy
\eqref{vNh}-\eqref{vNhi}. With this aim on $f_{n}\in\mathfrak{L}^{1}(\mathcal{H}_{n})$ we introduce
an analogue of the annihilation operator
\begin{eqnarray}\label{a}
    &&(\mathfrak{a}f)_{s}(1,\ldots,s)\doteq \mathrm{Tr}_{s+1}f_{s+1}(1,\ldots,s,s+1),\quad s\geq1,
\end{eqnarray}
and, therefore we have
\begin{eqnarray*}
     &&(e^{\pm\mathfrak{a}}f)_{s}(1,\ldots,s)=\sum\limits_{n=0}^{\infty}\frac{(\pm 1)^n}{n!}
        \mathrm{Tr}_{s+1,\ldots,{s+n}}f_{s+n}(1,\ldots,s+n).
\end{eqnarray*}
According to definition (\ref{Gexpg}) of the marginal correlation operators, i.e.
\begin{eqnarray*}
    &&G(t)=e^\mathfrak{a}g(t),
\end{eqnarray*}
where the sequence $g(t)$ is a solution of the von Neumann hierarchy for correlation operators
defined by group (\ref{rozvNh}), i.e. $g(t)=\mathcal{G}(t\mid g(0))$,
and to the equality: $g(0)=e^\mathfrak{-a}G(0)$, we finally derive
\begin{eqnarray}\label{s}
    &&G(t)=e^\mathfrak{a}\mathcal{G}(t\mid e^\mathfrak{-a}G(0)).
\end{eqnarray}

To set down formula (\ref{s}) in componentwise form we observe, that the following equality holds
\begin{eqnarray}\label{prod}
    &&\prod\limits_{X_i\subset \mathrm{P}}(e^{\mathfrak{-a}}G(0))_{|X_i|}(X_i)=
       \sum\limits_{k=0}^{\infty}\frac{(-1)^k}{k!}\mathrm{Tr}_{s+n+1,\ldots,s+n+k}
       \sum\limits_{k_1=0}^{k}\frac{k!}{k_{1}!(k-k_{1})!}\ldots\\
    &&\ldots\sum\limits_{k_{|\mathrm{P}|-1}=0}^{k_{|\mathrm{P}|-2}}
       \frac{k_{|\mathrm{P}|-2}!}{k_{|\mathrm{P}|-1}!(k_{|\mathrm{P}|-2}-k_{|\mathrm{P}|-1})!}
       G_{|X_1|+k-k_1}(0,X_1,s+n+1,\ldots,s+n+k-k_1)\ldots\nonumber\\
    &&\ldots G_{|X_{|\mathrm{P}|}|+k_{|\mathrm{P}|-1}}(0,X_{|\mathrm{P}|},s+n+k-k_{|\mathrm{P}|-1}+1,\ldots,s+n+k).\nonumber
\end{eqnarray}
Then according to formulas (\ref{s}) and (\ref{rozvNh}), for $s\geq1$ we have
\begin{eqnarray*}
    &&\hskip-8mmG_{s}(t,Y)=\sum\limits_{n=0}^{\infty}\frac{1}{n!}
        \,\mathrm{Tr}_{s+1,\ldots, s+n}\,\sum\limits_{\mathrm{P}:\,(1,\ldots, s+n)=\bigcup_i X_i}
        \mathfrak{A}_{|\mathrm{P}|}\big(t,\{X_1\},\ldots,\{X_{|\mathrm{P}|}\}\big)
        \prod\limits_{X_i\subset \mathrm{P}}(e^{\mathfrak{-a}}G(0))_{|X_i|}(X_i),
\end{eqnarray*}
where $\mathfrak{A}_{|\mathrm{P}|}(t)$ is $|\mathrm{P}|th$-order cumulant (\ref{cumulantP}),
and as a result of the validity of equality (\ref{prod}) for sequence (\ref{s}) we obtain
\begin{eqnarray*}\label{ss}
    &&\hskip-5mmG_{s}(t,1,\ldots,s)=\\
    &&\hskip-5mm=\sum\limits_{n=0}^{\infty}\frac{1}{n!}\,\mathrm{Tr}_{s+1,\ldots, s+n}\,
        \sum\limits_{k=0}^{n}(-1)^{k}\frac{n!}{k!(n-k)!}\,
        \sum\limits_{\mathrm{P}:\,(1,\ldots,s+n-k)=\bigcup_i X_i}
        \mathfrak{A}_{|\mathrm{P}|}\big(t,\{X_1\},\ldots,\{X_{|\mathrm{P}|}\}\big)\times\\
    &&\hskip-5mm\times \sum\limits_{k_1=0}^{k}\frac{k!}{k_{1}!(k-k_{1})!}\ldots
        \sum\limits_{k_{|\mathrm{P}|-1}=0}^{k_{|\mathrm{P}|-2}}
        \frac{k_{|\mathrm{P}|-2}!}{k_{|\mathrm{P}|-1}!(k_{|\mathrm{P}|-2}-k_{|\mathrm{P}|-1})!}\,
        G_{|X_1|+k-k_1}(0,X_1,\\
    &&\hskip-5mm s+n-k+1,\ldots,s+n-k_1)\ldots G_{|X_{|\mathrm{P}|}|+k_{|\mathrm{P}|-1}}(0,X_{|\mathrm{P}|},
        s+n-k_{|\mathrm{P}|-1}+1,\ldots,s+n).
\end{eqnarray*}

Consequently the solution expansion of the nonlinear quantum BBGKY hierarchy has the following structure
\begin{eqnarray}\label{sss}
    &&G_{s}(t,Y)=\sum\limits_{n=0}^{\infty}\frac{1}{n!}
        \,\mathrm{Tr}_{s+1,\ldots,s+n}\,U_{1+n}(t;\{Y\},s+1,\ldots,s+n\mid G(0)),\quad s\geq1,
\end{eqnarray}
where we introduce the notion of the $(1+n)th$-order reduced cumulant $U_{1+n}(t)$ of nonlinear groups
of operators (\ref{rozvNh})
\begin{eqnarray}\label{ssss}
    &&U_{1+n}(t;\{Y\},s+1,\ldots,s+n \mid G(0))\doteq\\
    &&\doteq\sum\limits_{k=0}^{n}(-1)^k \frac{n!}{k!(n-k)!}\,\sum\limits_{\mathrm{P}:\,
       (\theta(\{1,\ldots,s\}),s+1,\ldots,s+n-k)=\bigcup_i X_i}
       \mathfrak{A}_{|\mathrm{P}|}\big(t,\{X_1\},\ldots,\{X_{|\mathrm{P}|}\}\big)\times\nonumber\\
    &&\times \sum\limits_{k_1=0}^{k}\frac{k!}{k_{1}!(k-k_{1})!}\ldots
       \sum\limits_{k_{|\mathrm{P}|-1}=0}^{k_{|\mathrm{P}|-2}}
       \frac{k_{|\mathrm{P}|-2}!}{k_{|\mathrm{P}|-1}!(k_{|\mathrm{P}|-2}-k_{|\mathrm{P}|-1})!}
       G_{|X_1|+k-k_1}(0,X_1,\nonumber\\
    &&s+n-k+1,\ldots,s+n-k_1)\ldots G_{|X_{|\mathrm{P}|}|+k_{|\mathrm{P}|-1}}(0,X_{|\mathrm{P}|},
       s+n-k_{|\mathrm{P}|-1}+1,\ldots,s+n).\nonumber
\end{eqnarray}
We give simplest examples of reduced nonlinear cumulants (\ref{ssss}):
\begin{eqnarray*}
    &&U_{1}(t;\{Y\}\mid G(0))=\mathcal{G}(t;Y\mid G(0))=\\
    &&=\sum\limits_{\mathrm{P}:\,(1,\ldots,s)=
        \bigcup_i X_i}\mathfrak{A}_{|\mathrm{P}|}\big(t,\{X_1\},\ldots,\{X_{|\mathrm{P}|}\}\big)
        \prod\limits_{X_i\subset\mathrm{P}}G_{|X_i|}(0,X_{i}),\\
    &&U_{2}(t;\{Y\},s+1 \mid G(0))=\sum\limits_{\mathrm{P}:\,(Y,s+1)=\bigcup_i X_i}
        \mathfrak{A}_{|\mathrm{P}|}\big(t,\{X_1\},\ldots,\{X_{|\mathrm{P}|}\}\big)
        \prod\limits_{X_i\subset\mathrm{P}}G_{|X_i|}(0,X_{i})-\\
    &&-\sum\limits_{\mathrm{P}:\,(1,\ldots,s)=\bigcup_i X_i}
        \mathfrak{A}_{|\mathrm{P}|}\big(t,\{X_1\},\ldots,
        \{X_{|\mathrm{P}|}\}\big)\sum_{j=1}^{|\mathrm{P}|}
        \prod\limits_{\mbox{\scriptsize $\begin{array}{c}{X_i\subset\mathrm{P}},
        \\X_i\neq X_j\end{array}$}}G_{|X_i|}(0,X_{i})G_{|X_j|+1}(0,X_{j},s+1).
\end{eqnarray*}

We remark that in case of solution expansion (\ref{RozvBBGKY}) of the quantum BBGKY hierarchy, an
analog of reduced cumulant (\ref{ssss}) is the reduced cumulant of groups of operators (\ref{groupG})
defined by formula \cite{Pe95}
\begin{eqnarray*}
    &&U_{1+n}(t;\{Y\},s+1,\ldots,s+n)\doteq
       \sum\limits_{k=0}^{n}(-1)^k \frac{n!}{k!(n-k)!}\mathcal{G}_{s+n-k}(-t).
\end{eqnarray*}

\subsection{Reduced cumulants of nonlinear groups of operators}
We indicate some properties of reduced nonlinear cumulants (\ref{ssss}) of groups of operators
(\ref{rozvNh}). According to formula (\ref{sss}) and properties of cumulants (\ref{cumulantP}),
namely $\mathfrak{A}_n(0)=I\delta_{n,1}$, the following equality holds
\begin{eqnarray*}
    &&U_{1+n}(0;\{Y\},s+1,\ldots,s+n\mid G(0))=\\
    &&=\sum\limits_{k=0}^{n}(-1)^k \frac{n!}{k!(n-k)!}\,
       \mathfrak{A}_{1}\big(0,\{1,\ldots,s+n-k\}\big)G_{s+n}(0,1,\ldots,s+n)=\nonumber\\
    &&=G_{s+n}(0,1,\ldots,s+n)\delta_{n,0},\nonumber
\end{eqnarray*}
and hence the marginal correlation operators determined by series (\ref{sss}) satisfy initial
data (\ref{gBigfromDFBai}).

In case of $n=0$ for $f\in\mathfrak{L}^{1}_0(\mathcal{F}_\mathcal{H})$ in the sense of the norm
convergence of the space $\mathfrak{L}^{1}(\mathcal{H}_s)$ the infinitesimal generator of first-order
reduced cumulant (\ref{ssss}) coincides with generator (\ref{vNgenerator}) of the von Neumann hierarchy (\ref{vNh})
\begin{eqnarray*}
   &&\lim\limits_{t\rightarrow 0}\frac{1}{t}(U_{1}(t;\{Y\}\mid f)-f_{s}(Y))=
      \mathcal{N}(Y\mid f),\quad s\geq1,
\end{eqnarray*}
where the operator $\mathcal{N}(Y\mid f)$ is defined by formula (\ref{vNgenerator}). In case of $n=1$
for second-order reduced cumulant (\ref{ssss}) in the same sense we obtain the following equality
\begin{eqnarray*}
   &&\mathrm{Tr}_{s+1}\lim\limits_{t\rightarrow 0}\frac{1}{t}\,U_{2}(t;\{Y\},s+1\mid f)
     =\sum_{i\in Y}\mathrm{Tr}_{s+1}(-\mathcal{N}_{\mathrm{int}}(i,s+1))\big(f_{s+1}(t,Y,s+1)+\\
   &&+\sum_{\mbox{\scriptsize$\begin{array}{c}\mathrm{P}:(Y,s+1)=X_1\bigcup X_2,\\i\in
     X_1;s+1\in X_2\end{array}$}}f_{|X_1|}(t,X_1)f_{|X_2|}(t,X_2)\big),
\end{eqnarray*}
where notations are used as above for hierarchy (\ref{gBigfromDFBa}), and for $n\geq2$ as a consequence
of the fact that we consider a system of particles interacting by a two-body potential, it holds
\begin{eqnarray*}
  &&\mathrm{Tr}_{s+1,\ldots,s+n}\lim\limits_{t\rightarrow 0}\frac{1}{t}
     \,U_{1+n}(t;\{Y\},s+1,\ldots,s+n\mid f)=0.
\end{eqnarray*}

In case of initial data satisfying a chaos property, i.e. $G^{(1)}(0)\equiv(0,G_{1}(0,1),0,\ldots)$,
for the $(1+n)th$-order reduced cumulant we have
\begin{eqnarray*}
  &&U_{1+n}(t;\{Y\},s+1,\ldots,s+n \mid G^{(1)}(0))=
     \mathfrak{A}_{s+n}\big(t,1,\ldots,s+n\big)\prod\limits_{i=1}^{s+n}G_{1}(0,i),
\end{eqnarray*}
i.e. the only summand that gives contribution to the result is the one with $k=0$ and $|\mathrm{P}|=s+n$,
since otherwise there is at least one operator $G_s(0)$ with $s \geq 2$ in the last product. We note that
in section 4.1 the same result was obtained using the properties of solution of the von Neumann hierarchy (\ref{vNh}).

For the $(1+n)th$-order reduced cumulant (\ref{ssss}) the following inequality holds
\begin{eqnarray}\label{UMapping}
   &&\big\|U_{1+n}(t;\{Y\},s+1,\ldots,s+n \mid f)\big\|_{\mathfrak{L}^{1}(\mathcal{H}_{s+n})}\leq
      2n!s!(2e^{3})^{s+n}\emph{c}^{s+n},
\end{eqnarray}
where $\emph{c}\equiv\max\limits_{\mbox{\scriptsize $\begin{array}{c}\mathrm{P}:\,(1,\ldots,s+n-k)=
\bigcup_i X_i\end{array}$}}\max\limits_{\mbox{\scriptsize$\begin{array}{c}k,k_1,\ldots,k_{|\mathrm{P}|-1}\in
\\ \in(s+n-k+1,\ldots, s+n)\end{array}$}}\big(\|f_{|X_1|+k-k_1}\|_{\mathfrak{L}^{1}(\mathcal{H}_{|X_1|+k-k_1})},
\ldots\\ \ldots,\|f_{|X_{|\mathrm{P}|}|+k_{|\mathrm{P}|-1}}\|_{\mathfrak{L}^{1}(\mathcal{H}_{|X_{|\mathrm{P}|}|+
k_{|\mathrm{P}|-1}})}\big)$.

To prove this inequality we first remark that for cumulant (\ref{cumulantP}) the following estimate holds
\begin{eqnarray}\label{cumulantEstimate}
   &&\|\mathfrak{A}_{|\mathrm{P}|}(t,\{X_1\},\ldots,\{X_{|\mathrm{P}|}\})f_n\|_{\mathfrak{L}^{1}(\mathcal{H}_{n})}
      \leq|\mathrm{P}|! \,e^{|\mathrm{P}|}\|f_n\|_{\mathfrak{L}^{1}(\mathcal{H}_{n})}.
\end{eqnarray}
Indeed, we have
\begin{eqnarray*}
  &&\|\mathfrak{A}_{|\mathrm{P}|}(t,\{X_1\},\ldots,
     \{X_{|\mathrm{P}|}\})f_n\|_{\mathfrak{L}^{1}(\mathcal{H}_{n})}\leq\\
  &&\leq\sum\limits_{\mathrm{P}^{'}:\,(\{X_1\},\ldots,\{X_{|\mathrm{P}|}\})=
     \bigcup_k Z_k}({|\mathrm{P}^{'}|-1})!\|\prod\limits_{Z_k\subset\mathrm{P}^{'}}
     \mathcal{G}_{|\theta(Z_{k})|}(-t,\theta(Z_{k}))f_n\|_{\mathfrak{L}^{1}(\mathcal{H}_{n})}=\nonumber\\
  &&=\|f_n\|_{\mathfrak{L}^{1}(\mathcal{H}_{n})}\sum_{l=1}^{|\mathrm{P}|}\mathrm{s}(|\mathrm{P}|,l)(l-1)!,
\end{eqnarray*}
where $\mathrm{s}(|\mathrm{P}|,l)$ are the Stirling numbers of second kind and we use the isometric property
of the groups $\mathcal{G}_n(-t),\, n\geq1$. Estimate (\ref{cumulantEstimate}) holds as a consequence of
the inequality
\begin{eqnarray*}
    &&\sum_{l=1}^{|\mathrm{P}|} \mathrm{s}(|\mathrm{P}|,l)(l-1)!\leq|\mathrm{P}|!\,e^{|\mathrm{P}|}.
\end{eqnarray*}

Then owing to estimate (\ref{cumulantEstimate}), for the $(1+n)th$-order reduced cumulant (\ref{ssss})
we have
\begin{eqnarray*}
   &&\big\|U_{1+n}(t;\{Y\},s+1,\ldots,s+n\mid f)\big\|_{\mathfrak{L}^{1}(\mathcal{H}_{s+n})}\leq\\
   &&\leq\sum\limits_{k=0}^{n}\frac{n!}{k!(n-k)!}\,\sum\limits_{\mathrm{P}:\,(1,\ldots,s+n-k)=\bigcup_i X_i}
       |\mathrm{P}|!\,e^{|\mathrm{P}|}\sum\limits_{k_1=0}^{k}\frac{k!}{k_{1}!(k-k_{1})!}\ldots\nonumber\\
   &&\sum\limits_{k_{|\mathrm{P}|-1}=0}^{k_{|\mathrm{P}|-2}}
       \frac{k_{|\mathrm{P}|-2}!}{k_{|\mathrm{P}|-1}!(k_{|\mathrm{P}|-2}-k_{|\mathrm{P}|-1})!}
       \|f_{|X_1|+k-k_1}\|_{\mathfrak{L}^{1}(\mathcal{H}_{|X_1|+k-k_1})}
       \ldots\|f_{|X_{|\mathrm{P}|}|+k_{|\mathrm{P}|-1}}\|_{\mathfrak{L}^{1}(\mathcal{H}_{|X_{|\mathrm{P}|}|+
       k_{|\mathrm{P}|-1}})}\nonumber\leq\\
   &&\leq\sum\limits_{k=0}^{n}\frac{n!}{(n-k)!}\,\sum\limits_{\mathrm{P}:\,(1,\ldots,s+n-k)=\bigcup_i X_i}
       |\mathrm{P}|!\,e^{2|\mathrm{P}|-1}\emph{c}^{|\mathrm{P}|}.
\end{eqnarray*}
As result of using of the definition of the Stirling numbers of second kind $\mathrm{s}(s+n-k,l)$
and the inequalities
\begin{eqnarray*}
   &&\sum\limits_{k=0}^{n}\frac{n!}{(n-k)!}\,\sum\limits_{\mathrm{P}:\,(1,\ldots,s+n-k)=
      \bigcup_i X_i}|\mathrm{P}|!\,e^{2|\mathrm{P}|-1}=\sum\limits_{k=0}^{n}\frac{n!}{(n-k)!}\,
      \sum\limits_{l=1}^{s+n-k}\mathrm{s}(s+n-k,l)l!e^{2l-1}\leq\\
   &&\leq\sum\limits_{k=0}^{n}\frac{n!(s+n-k)!}{(n-k)!}e^{3(s+n-k)}\leq 2n!s!(2e^{3})^{s+n},
\end{eqnarray*}
we obtain estimate (\ref{UMapping}).

Thus, according to estimate (\ref{UMapping}), for initial data from the space
$\mathfrak{L}^{1}(\mathcal{H}_{n})$ series (\ref{sss}) converges under the condition:
$\mathfrak{c}\equiv\max_{n\geq1}\big\|G_n(0)\big\|_{\mathfrak{L}^{1}(\mathcal{H}_{n})}<(2e^{3})^{-1}$,
and the following inequality holds
\begin{eqnarray}\label{solEstimate}
   &&\big\|G_s(t)\big\|_{\mathfrak{L}^{1}(\mathcal{H}_{s})}\leq
      2s!(2e^{3}\mathfrak{c})^{s}\sum_{n=0}^{\infty}(2e^{3})^{n}\mathfrak{c}^{n}.
\end{eqnarray}

A solution of the Cauchy problem of the nonlinear quantum BBGKY hierarchy for marginal
correlation operators (\ref{gBigfromDFBa}) is determined by the following one-parametric mapping
\begin{eqnarray}\label{nhg}
   &&\mathbb{R}\ni t\rightarrow \mathcal{U}(t\mid f)=
      e^\mathfrak{a}\mathcal{G}(t\mid e^\mathfrak{-a}f),
\end{eqnarray}
which is defined on the space $\mathfrak{L}^{1}(\mathcal{F}_\mathcal{H})$ owing to estimate
(\ref{solEstimate}), and has the group property
\begin{eqnarray*}
   &&\mathcal{U}\big(t_1\mid \mathcal{U}(t_2\mid f)\big)=
      \mathcal{U}\big(t_2\mid \mathcal{U}(t_1\mid f)\big)=\mathcal{U}\big(t_1+t_2\mid f\big).
\end{eqnarray*}
Indeed, according to definition (\ref{a}) and taking into consideration the group property of the
mapping $\mathcal{G}(t\mid \cdot)$, we obtain
\begin{eqnarray*}
   &&\mathcal{U}\big(t_1+t_2\mid f\big)=e^\mathfrak{a}\mathcal{G}(t_1+t_2\mid e^\mathfrak{-a}f)=
      e^\mathfrak{a}\mathcal{G}(t_1\mid \mathcal{G}(t_2 \mid e^\mathfrak{-a}f)=\\
   &&=e^\mathfrak{a}\mathcal{G}(t_1\mid e^\mathfrak{-a}e^\mathfrak{a}
      \mathcal{G}(t_2\mid e^\mathfrak{-a}f) =
      e^\mathfrak{a}\mathcal{G}(t_1\mid e^\mathfrak{-a}\mathcal{U}\big(t_2\mid f\big))=
      \mathcal{U}\big(t_1\mid \mathcal{U}(t_2\mid f)\big).
\end{eqnarray*}

To construct the generator of the strong continuous group $\mathcal{U}(t;Y\mid\cdot)$ we differentiate
it in the sense of the norm convergence on the space $\mathfrak{L}^{1}(\mathcal{H}_s)$
\begin{eqnarray*}
   &&\frac{d}{dt}\mathcal{U}(t;Y\mid f)|_{t=0}=
      \frac{d}{dt}(e^\mathfrak{a}\mathcal{G}(t\mid e^\mathfrak{-a}f))_s(Y)|_{t=0}=\\
   &&=\sum\limits_{n=0}^{\infty}\frac{1}{n!}\,\mathrm{Tr}_{s+1,\ldots,s+n}\,
      \mathcal{N}(X\mid\mathcal{G}(t\mid e^\mathfrak{-a}f))|_{t=0}
      =(e^\mathfrak{a}\mathcal{N}(\cdot\mid e^{-\mathfrak{a}}f))_s(Y),
\end{eqnarray*}
where $\mathcal{N}(\cdot\mid f)$ is a generator of the von Neumann hierarchy (\ref{vNh}) defined
by formula (\ref{vNgenerator}) on the subspaces $\mathfrak{L}^{1}_{0}(\mathcal{H}_{s})\subset
\mathfrak{L}^{1}(\mathcal{H}_{s}),\,s\geq1$, or in the componentwise form
\begin{eqnarray}\label{gennh}
   &&(e^{\mathfrak{a}}\mathcal{N}(\cdot\mid e^{-\mathfrak{a}}f))_s(Y)=\mathcal{N}(Y\mid f)+
      \mathrm{Tr}_{s+1}\sum_{i\in Y}(-\mathcal{N}_{\mathrm{int}}(i,s+1))\big(f_{s+1}(Y,s+1)+\\
   &&+\sum_{\mbox{\scriptsize $\begin{array}{c}\mathrm{P}:(Y,s+1)=
      X_1\bigcup X_2,\\i\in X_1;s+1\in X_2\end{array}$}}f_{|X_1|}(X_1)f_{|X_2|}(X_2)\big),\nonumber
\end{eqnarray}
where we use notations as above for formula (\ref{gBigfromDFBa}), and transformations similar
to equalities (\ref{d1}) and (\ref{d2}) have been applied.

Indeed, to set down a generator of mapping (\ref{nhg}) in componentwise form we observe that
according to definitions (\ref{a}) and (\ref{vNgenerator}), the following equality holds
\begin{eqnarray*}
   &&(e^{\mathfrak{a}}\mathcal{N}(\cdot\mid e^{-\mathfrak{a}}f))_s(Y)=
       \sum\limits_{n=0}^{\infty}\frac{1}{n!}\,\mathrm{Tr}_{s+1,\ldots,s+n}\,
       \big(-\mathcal{N}_{s+n}(X)(e^\mathfrak{-a}f)_{s+n}(X)+\\
   &&+\sum\limits_{\mathrm{P}:\,X=X_{1}\bigcup X_2}\,\sum\limits_{i_{1}\in X_{1}}\sum\limits_{i_{2}\in X_{2}}
      (-\mathcal{N}_{\mathrm{int}}(i_{1},i_{2}))(e^\mathfrak{-a}f)_{|X_{1}|}(X_{1})
      (e^\mathfrak{-a}f)_{|X_{2}|}(X_{2})\big).
\end{eqnarray*}
Then in view of formulas (\ref{a}) and (\ref{prod}) we have
\begin{eqnarray*}
    &&\hskip-5mm(e^{\mathfrak{-a}}f)_{|X_1|}(X_1)(e^{\mathfrak{-a}}f)_{|X_2|}(X_2)=
       \sum\limits_{k=0}^{\infty}\frac{(-1)^k}{k!}\mathrm{Tr}_{s+n+1,\ldots,s+n+k}
       \sum\limits_{k_1=0}^{k}\frac{k!}{k_{1}!(k-k_{1})!}f_{|X_1|+k-k_1}(X_1,\\
    &&\hskip-5mms+n+1,\ldots,s+n+k-k_1)f_{|X_{2}|+k_{1}}(X_{2},s+n+k-k_{1}+1,\ldots,s+n+k),
\end{eqnarray*}
and as a result we obtain
\begin{eqnarray*}
   &&\hskip-5mm(e^{\mathfrak{a}}\mathcal{N}(\cdot\mid e^{-\mathfrak{a}}f))_s(1,\ldots,s)=\\
   &&\hskip-5mm=\sum\limits_{n=0}^{\infty}\frac{1}{n!}\,\mathrm{Tr}_{s+1,\ldots,s+n}
      \sum\limits_{k=0}^{n}(-1)^k\frac{n!}{k!(n-k)!}\,
      \big(-\mathcal{N}_{s+n-k}(1,\ldots,s+n-k)f_{s+n}(X)+\\
   &&\hskip-5mm+\sum\limits_{\mathrm{P}:\,(1,\ldots,s+n-k)=X_{1}\bigcup X_2}\,
      \sum\limits_{i_{1}\in X_{1}}\sum\limits_{i_{2}\in X_{2}}
      (-\mathcal{N}_{\mathrm{int}}(i_{1},i_{2}))\sum\limits_{k_1=0}^{k}\frac{k!}{k_{1}!(k-k_{1})!}\times\\
   &&\hskip-5mm\times f_{|X_1|+k-k_1}(X_1,s+n-k+1,\ldots,s+n-k_1)f_{|X_{2}|+k_{1}}(X_{2},s+n-k_{1}+1,\ldots,s+n)\big).
\end{eqnarray*}

Therefore the first term of this series is generator (\ref{vNgenerator}) of the von Neumann hierarchy
\begin{eqnarray*}
   &&I_1\equiv(-\mathcal{N}_s)f_s(Y)+\sum\limits_{\mathrm{P}:\,Y=X_{1}\bigcup X_2}
      \sum\limits_{i_{1}\in X_{1}}\sum\limits_{i_{2}\in X_{2}}
      (-\mathcal{N}_{\mathrm{int}}(i_1,i_2))f_{|X_1|}(X_1)f_{|X_2|}(X_2),
\end{eqnarray*}
and, as stated above, it coincides with the generator of the first order cumulant (\ref{ssss}).
For the second term of series (\ref{gennh}) we have
\begin{eqnarray*}
   &&I_2\equiv\mathrm{Tr}_{s+1}\big((-\mathcal{N}_{s+1}(Y,s+1))f_{s+1}-(-\mathcal{N}_s(Y))f_{s+1}+\\
   &&+\sum\limits_{\mathrm{P}:\,(Y,s+1)=X_{1}\bigcup X_2}\sum\limits_{i_{1}\in X_{1}}\sum\limits_{i_{2}
      \in X_{2}}(-\mathcal{N}_{\mathrm{int}}(i_1,i_2))f_{|X_1|}(X_1)f_{|X_2|}(X_2)-\\
   &&-\sum\limits_{\mathrm{P}:\,Y=X_{1}\bigcup X_2}\sum\limits_{i_{1}\in X_{1}}\sum\limits_{i_{2}\in
      X_{2}}(-\mathcal{N}_{\mathrm{int}}(i_1,i_2))(f_{|X_1|+1}(X_1,s+1)f_{|X_2|}(X_2)+\\
   &&+f_{|X_1|}(X_1)f_{|X_2|+1}(X_2,s+1))\big).
\end{eqnarray*}
Taking into account equality (\ref{sums}) in case of the set $(Y,s+1)$, it holds
\begin{eqnarray*}
   &&\hskip-7mm\sum\limits_{\mathrm{P}:\,(Y,s+1)=X_{1}\bigcup X_2}\sum\limits_{i_{1}\in X_{1}}\sum\limits_{i_{2}
      \in X_{2}}(-\mathcal{N}_{\mathrm{int}}(i_1,i_2))f_{|X_1|}(X_1)f_{|X_2|}(X_2)=\\
   &&\hskip-7mm=\sum\limits_{\mathrm{P}:\,(Y,s+1)=X_{1}\bigcup X_2}\sum\limits_{i_{1}\in X_{1}\bigcap Y}
      \sum\limits_{i_{2}\in X_{2}\bigcap Y}(-\mathcal{N}_{\mathrm{int}}(i_1,i_2))f_{|X_1|}(X_1)f_{|X_2|}(X_2)+\\
   &&\hskip-7mm+\sum\limits_{\mbox{\scriptsize$\begin{array}{c}\mathrm{P}:\,(Y,s+1)=X_{1}\bigcup X_2,\\
      s+1\in X_2 \end{array}$}}\sum\limits_{i_{1}\in X_{1}\bigcap Y}
      (-\mathcal{N}_{\mathrm{int}}(i_1,s+1))f_{|X_1|}(X_1)f_{|X_2|}(X_2)=\\
   &&\hskip-7mm=\sum\limits_{\mathrm{P}:\,Y=Y_{1}\bigcup Y_2}\sum\limits_{i_{1}\in Y_{1}}
      \sum\limits_{i_{2}\in Y_{2}}(-\mathcal{N}_{\mathrm{int}}(i_1,i_2))(f_{|Y_1|+1}(Y_1,s+1)f_{|Y_2|}(Y_2)+\\
   &&\hskip-7mm+f_{|Y_1|}(Y_1)f_{|Y_2|+1}(Y_2,s+1))+\sum\limits_{i\in Y}(-\mathcal{N}_{\mathrm{int}}(i,s+1))
      \hskip-2mm\sum\limits_{\mbox{\scriptsize$\begin{array}{c}\mathrm{P}:\,(Y,s+1)=X_{1}\bigcup X_2,\\
      i\in X_1;s+1\in X_2 \end{array}$}}f_{|X_1|}(X_1)f_{|X_2|}(X_2),
\end{eqnarray*}
where $\sum_{\mbox{\scriptsize$\begin{array}{c}\mathrm{P}:(Y,s+1)=X_1\bigcup X_2,\\s+1\in X_2\end{array}$}}$
is the sum over all possible partitions of the set $(Y,s+1)$ into two mutually disjoint subsets $X_1$ and
$X_2$ such that $(s+1)th$ particle index belongs to set $X_2$. As a result we obtain
\begin{eqnarray*}
   &&I_2=\mathrm{Tr}_{s+1}\big(\sum\limits_{i\in Y}(-\mathcal{N}_{\mathrm{int}}(i,s+1))f_{s+1}+\\
   &&+\sum\limits_{i\in Y}(-\mathcal{N}_{\mathrm{int}}(i,s+1))
      \sum\limits_{\mbox{\scriptsize$\begin{array}{c}\mathrm{P}:\,(Y,s+1)=X_{1}\bigcup X_2,\\
      i\in X_1; s+1 \in X_2 \end{array}$}}f_{|X_1|}(X_1)f_{|X_2|}(X_2)\big),
\end{eqnarray*}
i.e. this term coincides with the generator of second-order cumulant (\ref{ssss}).

In case of a two-body interaction potential other terms of series (\ref{gennh}) are
identically equal to zero. This statement is a consequence of the structure of expansion
(\ref{gennh}) and of the fact that its third term equals zero. Indeed as a result of
regrouping terms in the expression of the third term we obtain
\begin{eqnarray*}
   &&I_3=\frac{1}{2!}\mathrm{Tr}_{s+1,s+2}\big((-\mathcal{N}_{s+2}(Y,s+1,s+2)-
     2!(-\mathcal{N}_{s+1}(Y,s+1))+(-\mathcal{N}_{s}(Y)))f_{s+2}+\\
   &&+\sum\limits_{\mathrm{P}:\,(Y,s+1,s+2)=X_{1}\bigcup X_2}
     \sum\limits_{i_{1}\in X_{1}\bigcap Y}\sum\limits_{i_{2}\in X_{2}\bigcap
     Y}(-\mathcal{N}_{\mathrm{int}}(i_1,i_2))f_{|X_1|}(X_1)f_{|X_2|}(X_2)-\\
   &&-2!\sum\limits_{\mathrm{P}:\,(Y,s+1)=X_{1}\bigcup X_2}\sum\limits_{i_{1}\in X_{1}\bigcap Y}
     \sum\limits_{i_{2}\in X_{2}\bigcap Y}(-\mathcal{N}_{\mathrm{int}}(i_1,i_2))
     (f_{|X_1|+1}(X_1,s+1)f_{|X_2|}(X_2)+\\
   &&+f_{|X_1|}(X_1)f_{|X_2|+1}(X_2,s+1))+\\
   &&+\sum\limits_{\mathrm{P}:\,Y=X_{1}\bigcup X_2}\sum\limits_{i_{1}\in X_{1}\bigcap Y}
     \sum\limits_{i_{2}\in X_{2}\bigcap Y}(-\mathcal{N}_{\mathrm{int}}(i_1,i_2))
     (f_{|X_1|+2}(X_1,s+1,s+2)f_{|X_2|}(X_2)+\\
   &&+f_{|X_1|}(X_1)f_{|X_2|+2}(X_2,s+1,s+2)+2!f_{|X_1|+1}(X_1,s+1)f_{|X_2|+1}(X_2,s+2))\big)=0.
\end{eqnarray*}

Thus, we conclude the validity of formula (\ref{gennh}) in case of a two-body interaction potential
which describes the structure of the infinitesimal generator of mapping (\ref{nhg}) in the general case.

\subsection{An existence theorem}
For an abstract initial-value problem of hierarchy (\ref{gBigfromDFBa}) in the space
$\mathfrak{L}^{1}(\mathcal{F}_\mathcal{H})$ the following theorem is true.
\begin{theorem}
If $\max_{n\geq1}\big\|G_n(0)\big\|_{\mathfrak{L}^{1}(\mathcal{H}_{n})}<(2e^{3})^{-1}$,
then in case of bounded interaction potentials for $t\in\mathbb{R}$ a solution of the
Cauchy problem of the nonlinear quantum BBGKY hierarchy  (\ref{gBigfromDFBa})-(\ref{gBigfromDFBai})
is determined by expansion (\ref{sss}).
If $G_{n}(0)\in\mathfrak{L}^{1}_{0}(\mathcal{H}_{n})\subset\mathfrak{L}^{1}(\mathcal{H}_{n})$, it is a
strong (classical) solution and for arbitrary initial data $G_{n}(0)\in\mathfrak{L}^{1}(\mathcal{H}_{n})$
it is a weak (generalized) solution.
\end{theorem}
\begin{proof}
It will be recalled that according to estimate (\ref{solEstimate}), series (\ref{sss}) converges
under the condition: $\max_{n\geq1}\big\|G_n(0)\big\|_{\mathfrak{L}^{1}(\mathcal{H}_{n})}<(2e^{3})^{-1}$.
To prove that a strong solution of the nonlinear BBGKY hierarchy (\ref{gBigfromDFBa}) is given by
expansion (\ref{sss}) we first differentiate it over time variable in the sense of a pointwise
convergence on the space $\mathfrak{L}^{1}(\mathcal{H}_{n})$, i.e. for every function from the domain
$\psi_{s}\in\mathcal{D}(H_{s})\subset\mathcal{H}_{s}$. Taking into account the group property of mapping
(\ref{nhg}) generated by expansion (\ref{sss}) and properties of reduced nonlinear cumulants (\ref{ssss})
of groups of operators (\ref{rozvNh}), for
$G_{n}(0)\in\mathfrak{L}^{1}_{0}(\mathcal{H}_{n})\subset\mathfrak{L}^{1}(\mathcal{H}_{n}),\, n\geq1$,
we obtain
\begin{eqnarray}\label{pwcsh}
   &&\lim\limits_{\Delta t\rightarrow 0}\frac{1}{\Delta t}\big(\sum\limits_{n=0}^{\infty}\frac{1}{n!}
        \,\mathrm{Tr}_{s+1,\ldots,s+n}(U_{1+n}(\Delta t;\{Y\},X\setminus Y\mid G(t))-G_{s}(t,Y)\big)\psi_{s}=\\
   &&=\mathcal{N}(Y\mid G(t))\psi_{s}+
      \mathrm{Tr}_{s+1}\sum_{i\in Y}(-\mathcal{N}_{\mathrm{int}}(i,s+1))\big(G_{s+1}(t,Y,s+1)+\nonumber\\
   &&+\sum_{\mbox{\scriptsize$\begin{array}{c}\mathrm{P}:(Y,s+1)=X_1\bigcup X_2,\\i\in
      X_1;s+1\in X_2\end{array}$}}G_{|X_1|}(t,X_1)G_{|X_2|}(t,X_2)\big)\psi_{s},\nonumber
\end{eqnarray}
where expansion (\ref{sss}) is denoted by the symbol $G_{s}(t,Y)$. Since $G_{n}(0)\in\mathfrak{L}_{0}^{1}(\mathcal{H}_{n})\subset\mathfrak{L}^{1}(\mathcal{H}_{n}),\,n\geq1$,
then using equality (\ref{pwcsh}), in the sense of the norm convergence in
$\mathfrak{L}^{1}(\mathcal{H}_{n})$ we finally establish the validity of the equality
\begin{eqnarray*}
   &&\hskip-7mm\lim_{\Delta t\rightarrow 0}\mathrm{Tr}_{1,\ldots,s}
        \big|\frac{1}{\Delta t}\big(\sum\limits_{n=0}^{\infty}\frac{1}{n!}\,
        \mathrm{Tr}_{s+1,\ldots,s+n}U_{1+n}(t+\Delta t;\{Y\},X\setminus Y \mid G(0))-\\
   &&\hskip-7mm-\sum\limits_{n=0}^{\infty}\frac{1}{n!}\,
        \mathrm{Tr}_{s+1,\ldots,s+n}U_{1+n}(t;\{Y\},X\setminus Y\mid G(0))\big)-\\
   &&\hskip-7mm-\Big(\mathcal{N}(Y\mid G(t))+
        \mathrm{Tr}_{s+1}\sum_{i\in Y}(-\mathcal{N}_{\mathrm{int}}(i,s+1))\big(G_{s+1}(t,Y,s+1)+\\
   &&\hskip-7mm+\sum_{\mbox{\scriptsize$\begin{array}{c}\mathrm{P}:(Y,s+1)=X_1\bigcup X_2,\\i\in
      X_1;s+1\in X_2\end{array}$}}G_{|X_1|}(t,X_1)G_{|X_2|}(t,X_2)\big)\Big)\big|=0,
\end{eqnarray*}
which means that a strong solution of the nonlinear BBGKY hierarchy (\ref{gBigfromDFBa}) is given by
expansion (\ref{sss}) in case of initial data from the subspaces
$\mathfrak{L}^{1}_{0}(\mathcal{H}_{n})\subset\mathfrak{L}^{1}(\mathcal{H}_{n}),\, n\geq1$.

Let us give a sketch of the prove that in case of arbitrary initial data 
$G_{n}(0)\in\mathfrak{L}^{1}(\mathcal{H}_{n}),\,n\geq1$, expansion (\ref{sss}) is a weak solution of the 
initial-value problem (\ref{gBigfromDFBa})-(\ref{gBigfromDFBai}). To this end we introduce the functional
\begin{eqnarray}\label{functional}
   &&(f,G(t))\doteq\sum_{s=0}^{\infty}\frac{1}{s!}\,\mathrm{Tr}_{1,\ldots,s}f_s(Y)G_s(t,Y),
\end{eqnarray}
where $f=(0,f_1,\ldots,f_n,\ldots)\in\mathfrak{L}_{0}(\mathcal{F}_{\mathcal{H}})$ is the finite sequence
of degenerate bounded operators with infinitely times differentiable kernels with compact supports. For
$G_{n}(0)\in\mathfrak{L}^{1}(\mathcal{H}_{n})$ and $f_{n}\in\mathfrak{L}_{0}(\mathcal{H}_{n})$ functional
(\ref{functional}) exists.

We transform functional (\ref{functional}) to the following form
\begin{eqnarray}\label{trf}
  &&(f,G(t))=(f,e^\mathfrak{a}\mathcal{G}(t\mid e^\mathfrak{-a}G(0)))=
     (e^\mathfrak{a^{+}}f,\mathcal{G}(t\mid e^\mathfrak{-a}G(0))),
\end{eqnarray}
where the operator $\mathfrak{a}$ is defined by (\ref{a}) and on $f_s\in \mathfrak{L}_{0}(\mathcal{H}_{s})$
the operator $\mathfrak{a^{+}}$ is defined by the formula (an analog of the creation operator)
\begin{eqnarray*}\label{opercr}
   &&(\mathfrak{a}^{+}f)_{s}(Y)\doteq\sum_{j=1}^s\,f_{s-1}(Y\setminus(j)).
\end{eqnarray*}

To differentiate obtained functional (\ref{trf}) with respect to the time variable we use the corresponding
result \cite{GP} of the differentiation of group (\ref{rozvNh}) of the von Neumann hierarchy (\ref{vNh}).
As a result we derive that
\begin{eqnarray*}
   &&\frac{d}{dt}(f,G(t))=\sum_{s=0}^{\infty}\frac{1}{s!}\,\mathrm{Tr}_{1,\ldots,s}
     \big(\mathcal{N}_{s}(Y)(e^\mathfrak{a^+}f)_s(Y)\mathcal{G}(t,Y\mid e^\mathfrak{-a}G(0))+\\
   &&+\sum\limits_{\mathrm{P}:Y=X_{1}\bigcup X_2}\sum\limits_{i_{1}\in X_{1}}\sum\limits_{i_{2}\in X_{2}}
     \mathcal{N}_{\mathrm{int}}(i_{1},i_{2})(e^\mathfrak{a^+}f)_s(Y)
     \mathcal{G}(t,X_{1}\mid e^\mathfrak{-a}G(0))\mathcal{G}(t,X_{2}\mid e^\mathfrak{-a}G(0))\big).
\end{eqnarray*}
Taking into account the structure of expansion (\ref{sss}) of the nonlinear quantum BBGKY hierarchy solution,
for $f_s\in \mathfrak{L}_{0}(\mathcal{H}_{s}),\, s\geq1$, the following equality holds
\begin{eqnarray*}\label{wh}
   &&\hskip-8mm\frac{d}{dt}(f,G(t))=\sum_{s=0}^{\infty}\frac{1}{s!}\,\mathrm{Tr}_{1,\ldots,s}
      \Big(\big(\mathcal{N}_{s}(Y)f_s(Y)+\sum_{\mbox{\scriptsize$\begin{array}{c}i,j=1\\i\neq j\end{array}$}}^{s}
      \mathcal{N}_{\mathrm{int}}(i,j)f_{s-1}(Y\setminus(j))\big)G_{s}(t,Y)+\\
   &&\hskip-8mm+\sum\limits_{\mathrm{P}:Y=X_{1}\bigcup X_2}\sum\limits_{i_{1}\in X_{1}}
      \sum\limits_{i_{1}\in X_{1}}\sum\limits_{i_{2}\in X_{2}}
      \mathcal{N}_{\mathrm{int}}(i_{1},i_{2})f_s(Y)G_{|X_{1}|}(t,X_{1})G_{|X_{2}|}(t,X_{2})+\nonumber\\
   &&\hskip-8mm+\sum_{\mbox{\scriptsize$\begin{array}{c}i,j=1\\i\neq j\end{array}$}}^{s}
      \mathcal{N}_{\mathrm{int}}(i,j)f_{s-1}(Y\setminus(j))\sum_{\mbox{\scriptsize$\begin{array}{c}\mathrm{P}:Y=
      X_1\bigcup X_2\\i\in X_1;j\in X_2\end{array}$}}G_{|X_1|}(t,X_1)G_{|X_2|}(t,X_2)\big)\Big).\nonumber
\end{eqnarray*}
This equation means that in case of arbitrary initial data $G_{n}(0)\in\mathfrak{L}^{1}(\mathcal{H}_{n}),\,n\geq1$,
a weak solution of the initial-value problem (\ref{gBigfromDFBa})-(\ref{gBigfromDFBai}) is given by expansion (\ref{sss}).

\end{proof}

\subsection{Remark: the nonlinear Vlasov hierarchy}
We give comments on the mean field asymptotic behavior \cite{Sp80} of constructed
solution (\ref{sss}).

Let us suppose the existence of the mean field limit of initial state in the following
sense
\begin{eqnarray}\label{asic}
   &&\lim\limits_{\epsilon\rightarrow 0}\big\|\epsilon^{s}G_{n}(0)-
      g_{n}(0)\big\|_{\mathfrak{L}^{1}(\mathcal{H}_n)}=0, \quad n\geq1.
\end{eqnarray}
Then there exists the mean field limit $g_s(t,1,\ldots,s),\,s\geq1$, of marginal correlation
operators (\ref{sss})
\begin{eqnarray*}\label{asymp}
   &&\lim\limits_{\epsilon\rightarrow 0}\big\|\epsilon^{s}G_{s}(t)-
      g_{s}(t)\big\|_{\mathfrak{L}^{1}(\mathcal{H}_s)}=0, \quad s\geq1,
\end{eqnarray*}
which is governed by the nonlinear Vlasov quantum hierarchy
\begin{eqnarray}\label{gBigfromDFBa_lim}
   &&\frac{d}{dt}g_s(t,Y)=\sum_{i\in Y}(-\mathcal{N}(i))g_{s}(t,Y) +
      \mathrm{Tr}_{s+1}\sum_{i\in Y}(-\mathcal{N}_{\mathrm{int}}(i,s+1))\times\\
   &&\times\big(g_{s+1}(t,Y,s+1)+\sum_{\mbox{\scriptsize
      $\begin{array}{c}\mathrm{P}:(Y,s+1)=X_1\bigcup X_2,\\i\in X_1;s+1\in X_2\end{array}$}}
      g_{|X_1|}(t,X_1)g_{|X_2|}(t,X_2)\big),\quad s\geq1, \nonumber
\end{eqnarray}
where notations similar to hierarhy (\ref{gBigfromDFBa}) are used.

If initial data satisfies chaos property, then we establish
\begin{eqnarray}\label{Gcid}
   &&\lim\limits_{\epsilon\rightarrow 0}\big\|\epsilon^{s}G_{s}(t)
      \big\|_{\mathfrak{L}^{1}(\mathcal{H}_s)}=0,\quad s\geq2,
\end{eqnarray}
since solution expansions (\ref{GUG(0)}) for marginal correlation operators are defined by the
$(s+n)th$-order cumulants as contrasted to solution expansions (\ref{RozvBBGKY}) for marginal
density operators defined by the $(1+n)th$-order cumulants and in the consequence of the
following formula on an asymptotic perturbation of cumulants of groups of operators \cite{Kato}
\begin{eqnarray*}
   &&\lim\limits_{\epsilon\rightarrow0}\big\|\frac{1}{\epsilon^{n}}\,
     \mathfrak{A}_{s+n}(t,1,\ldots,s+n)f_{s+n}\big\|_{\mathfrak{L}^{1}(\mathcal{H}_{s+n})}=0.
\end{eqnarray*}
In case of $s=1$ provided that (\ref{asic}) we have
\begin{eqnarray*}
   &&\lim\limits_{\epsilon\rightarrow 0}\big\|\epsilon G_{1}(t)-
     g_{1}(t)\big\|_{\mathfrak{L}^{1}(\mathcal{H})}=0,
\end{eqnarray*}
where for finite time interval the limit one-particle marginal correlation operator $g_1(t,1)$
is given by the norm convergent on the space $\mathfrak{L}^{1}(\mathcal{H})$ series
\begin{eqnarray}\label{1mco}
   &&\hskip-10mm g_{1}(t,1)=\\
   &&\hskip-10mm =\sum\limits_{n=0}^{\infty}\int\limits_0^tdt_{1}\ldots
      \int\limits_0^{t_{n-1}}dt_{n}\,\mathrm{Tr}_{2,\ldots,n+1}\mathcal{G}_{1}(-t+t_{1},1)
      (-\mathcal{N}_{\mathrm{int}}(1,2))\prod\limits_{j_1=1}^{2}
      \mathcal{G}_{1}(-t_{1}+t_{2},j_1)\ldots\nonumber\\
   &&\hskip-10mm \ldots\prod\limits_{i_{n}=1}^{n}\mathcal{G}_{1}(-t_{n}+t_{n},i_{n})
      \sum\limits_{k_{n}=1}^{n}(-\mathcal{N}_{\mathrm{int}}(k_{n},n+1))\prod\limits_{j_n=1}^{n+1}
      \mathcal{G}_{1}(-t_{n},j_n)\prod\limits_{i=1}^{n+1}g_1(0,i),\nonumber
\end{eqnarray}
which obviously coincides with iteration series of the Vlasov quantum kinetic equation \cite{G11}.
For bounded interaction potential (\ref{H_n}) series (\ref{1mco}) is norm convergent on the space
$\mathfrak{L}^{1}(\mathcal{H})$ under the condition: $t<t_0\equiv\big(2\,
\|\Phi\|_{\mathfrak{L}(\mathcal{H}_{2})}\|g_1(0)\|_{\mathfrak{L}^{1}(\mathcal{H})}\big)^{-1}$.

In view of the validity of limit (\ref{Gcid}) from the Vlasov nonlinear quantum hierarchy
(\ref{gBigfromDFBa_lim}) we also conclude that limit one-particle marginal correlation operator
(\ref{1mco}) is governed by the Cauchy problem of the Vlasov quantum kinetic equation
\begin{eqnarray}\label{Vlasov1}
  &&\frac{d}{dt}g_{1}(t,1)=-\mathcal{N}(1)g_{1}(t,1)+
     \mathrm{Tr}_{2}(-\mathcal{N}_{\mathrm{int}}(1,2))g_{1}(t,1)g_{1}(t,2),
\end{eqnarray}
and consequently for pure states we derive the Hartree equation.

Thus, the nonlinear Vlasov quantum hierarchy (\ref{gBigfromDFBa_lim}) describes the evolution
of initial correlations.

%%%%%%%%%%%%%%%%%%%%%%%%%%%%%%%%%%%%%%%%%%%%%%%%%%%%%%%%%%%%%%%%%%%%%%%%%%%%%%%%%%%%%%%%%%%%%%%%%%%%%%%%%%%%%%%%%%%%%%%%%%%%%%
%%%%%%%%%%%%%%%%%%%%%%%%%%%%%%%%%%%%%%%%%%%%%%%%%%%%%%%%%%%%%%%%%%%%%%%%%%%%%%%%%%%%%%%%%%%%%%%%%%%%%%%%%%%%%%%%%%%%%%%%%%%%%%

\section{Conclusion}
In the paper the origin of the microscopic description of non-equilibrium correlations of
quantum many-particle systems obeying the Maxwell-Boltzmann statistics has been considered.
The nonlinear quantum BBGKY hierarchy (\ref{gBigfromDFBa}) for marginal correlation operators
was introduced. It gives an alternative approach to the description of the state evolution
of quantum infinite-particle systems in comparison with quantum BBGKY hierarchy for marginal
density operators \cite{BogLect,Pe95}. The evolution of both finitely and infinitely many
quantum particles is described by initial-value problem of the nonlinear quantum BBGKY hierarchy
(\ref{gBigfromDFBa}) and in case of finitely many particles the nonlinear quantum BBGKY hierarchy
is equivalent to the von Neumann hierarchy (\ref{vNh}).

A nonperturbative solution of the nonlinear quantum BBGKY hierarchy is constructed in the form
of expansion (\ref{sss}) over particle clusters which evolution is governed by corresponding-order
cumulant (\ref{ssss}) of the nonlinear groups of operators generated by solution (\ref{rozvNh}) of
the von Neumann hierarchy (\ref{vNh}). We established that in case of absence of correlations at
initial time the correlations generated by the dynamics of quantum many-particle systems (\ref{GUG(0)})
are completely determined by cumulants (\ref{cumulantP}) of groups of operators (\ref{groupG}).

Thus, the cumulant structure of solution (\ref{rozvNh}) of the von Neumann hierarchy (\ref{vNh})
induces the cumulant structure of solution expansion (\ref{sss}) of initial-value problem of the
nonlinear quantum BBGKY hierarchy (\ref{gBigfromDFBa}).

We emphasize that intensional Banach spaces for the description of states of infinite-particle
systems, which are suitable for the description of the kinetic evolution or equilibrium states,
are different from the exploit spaces \cite{Pe95},\cite{CGP97}. Therefore marginal correlation
operators from the space of trace-class operators describe finitely many quantum particles. In
order to describe the evolution of infinitely many particles we have to construct solutions for
initial data from more general Banach spaces than the space of sequences of trace-class operators.
For example, it can be the space of sequences of bounded translation invariant operators which
contains the marginal density operators of equilibrium states \cite{Gen}. In that case every term
of the solution expansion of the nonlinear quantum BBGKY hierarchy (\ref{sss}) contains the divergent
traces, which can be renormalized due to the cumulant structure of solution expansion (\ref{ssss}).

The mean field asymptotic behavior of constructed solution (\ref{sss}) is governed by the nonlinear
Vlasov quantum hierarchy (\ref{gBigfromDFBa_lim}). In such approximation this hierarchy describes
the evolution of initial correlations and in case of its absence the nonlinear Vlasov hierarchy
(\ref{gBigfromDFBa_lim}) is equivalent to the Vlasov quantum kinetic equation (\ref{Vlasov1}).

Following to the paper \cite{GP} the obtained results can be also generalized on many-particle
systems obeying the Fermi-Dirac and Bose-Einstein statistics (\ref{eqG}).

%%%%%%%%%%%%%%%%%%%%%%%%%%%%%%%%%%%%%%%%%%%%%%%%%%%%%%%%%%%%%%%%%%%%%%%%%%%%%%%%%%%%%%%%%%%%%%%%%%%%%%%%%%%%%%%%%%%%%%%%%%%%%%
%%%%%%%%%%%%%%%%%%%%%%%%%%%%%%%%%%%%%%%%%%%%%%%%%%%%%%%%%%%%%%%%%%%%%%%%%%%%%%%%%%%%%%%%%%%%%%%%%%%%%%%%%%%%%%%%%%%%%%%%%%%%%%
\bigskip

\addcontentsline{toc}{section}{References}
\renewcommand{\refname}{References}

\end{document}